\documentclass[acmsmall,screen,nonacm]{acmart}

\renewcommand\footnotetextcopyrightpermission[1]{}
\usepackage{mathtools}

\usepackage{listings}
\usepackage[ruled,vlined,linesnumbered]{algorithm2e}
\SetKwProg{Fn}{function}{:}{}
\SetKwFunction{TABLED}{Tabled}
\SetKwFunction{TABLEDWHNF}{TabledWHNF}
\SetKwFunction{RESOLVE}{Resolve}
\SetKwFunction{WHNF}{WHNF}
\SetKwFunction{Subst}{Subst}
\SetKw{Break}{break}

\usepackage{tikz}
\usetikzlibrary{arrows.meta, positioning, fit}

\usepackage{latexml}
\iflatexml
\newenvironment{autobreak}{}{}
\else
\usepackage{autobreak}
\fi

\providecommand{\anon}[2][]{#2}

\newcommand{\Lamsh}{\mathsf{Lam}}
\newcommand{\Appsh}{\mathsf{App}}
\newcommand{\Varsh}{\mathsf{Var}}

\AtBeginDocument{%
  \ifTablambdaLNCS\else
  \theoremstyle{acmdefinition}%
  \fi
}

\makeatletter
\@ifundefined{ifTablambdaBody}{\newif\ifTablambdaBody\TablambdaBodytrue}{}
\@ifundefined{ifTablambdaAppendix}{\newif\ifTablambdaAppendix\TablambdaAppendixtrue}{}
\@ifundefined{ifPaperChinese}{\newif\ifPaperChinese\PaperChinesefalse}{}
\@ifundefined{ifTablambdaLNCS}{\newif\ifTablambdaLNCS\TablambdaLNCSfalse}{}
\makeatother

\newcommand{\bilingual}[2]{#1}
\newcommand{\tablambdaKeywords}[1]{\bilingual{tabling#1 weak-head reduction#1 cyclic graphs#1 memoization#1 least fixpoint#1
    rational trees#1 structure as identity#1 Datalog#1 program analysis#1 domain-specific language#1
    \texorpdfstring{$\lambda$-calculus}{lambda-calculus}}{tabling#1 弱头归约#1 环图#1 记忆化#1 最小不动点#1
    有理树#1 结构即同一性#1 Datalog#1 程序分析#1 领域专用语言#1
    \texorpdfstring{$\lambda$-演算}{lambda-calculus}}}
\ifPaperChinese
\usepackage{CJKutf8}
\renewcommand{\bilingual}[2]{\texorpdfstring{#2}{#1}}
\makeatletter
\let\TablambdaParagraphFont\@parfont
\def\@parfont{\normalfont\TablambdaParagraphFont}
\makeatother
\fi

\ifTablambdaLNCS\else
\acmJournal{PACMPL}
\fi

\begin{document}
\ifPaperChinese
\begin{CJK*}{UTF8}{gbsn}
  \renewcommand{\proofname}{证明}
  \ifTablambdaLNCS\else
  \makeatletter
  \newcommand{\tablambdaThmName}[1]{\@ifundefined{tablambda@thm@#1}{#1}{\csname tablambda@thm@#1\endcsname}}
  \renewcommand{\thmname}[1]{\tablambdaThmName{#1}}
  \@namedef{tablambda@thm@Theorem}{定理}
  \@namedef{tablambda@thm@Lemma}{引理}
  \@namedef{tablambda@thm@Corollary}{推论}
  \@namedef{tablambda@thm@Proposition}{命题}
  \@namedef{tablambda@thm@Conjecture}{猜想}
  \@namedef{tablambda@thm@Definition}{定义}
  \@namedef{tablambda@thm@Example}{例}
  \@namedef{tablambda@thm@Remark}{注}
  \makeatother
  \fi
  \fi

  \ifTablambdaLNCS
  \title{\bilingual{Cyclic Graphs and Memoization in Pure $\lambda$-Calculus}{纯 $\lambda$-演算中的环图与记忆化}}
  \titlerunning{Cyclic Graphs and Memoization in Pure $\lambda$-Calculus}
  \else
  \title[Cyclic Graphs and Memoization in Pure $\lambda$-Calculus]{\bilingual{Cyclic Graphs and Memoization in Pure $\lambda$-Calculus}{纯 $\lambda$-演算中的环图与记忆化}}
  \fi
  \ifTablambdaBody\else
  \subtitle{Supplementary Material}
  \fi

  \ifTablambdaLNCS
  \author{Bo Yang\orcidID{0000-0003-2757-9115}\thanks{The author completed this work independently and outside the scope of employment at Figure AI. No Figure AI time, funding, or other resources were used in support of this work.}}
  \authorrunning{B. Yang}
  \institute{Figure AI Inc., San Jose, California, USA\\ \email{yang-bo@yang-bo.com}}
  \else
  \author{Bo Yang}
  \orcid{0000-0003-2757-9115}
  \affiliation{%
    \institution{Figure AI Inc.}
    \city{San Jose}
    \state{California}
    \country{USA}
  }
  \email{yang-bo@yang-bo.com}
  \thanks{The author completed this work independently and outside the scope of employment at Figure AI. No Figure AI time, funding, or other resources were used in support of this work.}
  \fi

  \ifTablambdaLNCS\else

  \fi

  \ifTablambdaLNCS
  \maketitle
  \fi
  \ifTablambdaBody
  \begin{abstract}
    \ifPaperChinese
    纯函数式编程把不可变与非严格求值作为默认;我们更进一步,把内化(interning)与记忆化作为计算的默认。在以往对纯 $\lambda$-演算的操作语义中,表示并变换环状与无限数据需要额外的递归构造,如 $\texttt{letrec}$、$\mu$-绑定子,或为图归约内置的 $Y$;而要共享一个记忆化函数或动态规划函数中重复的计算,又需要一个非纯的缓存。我们证明无需任何此类扩展。我们把 tabling,即求解最小不动点方程的标准方法,应用于弱头归约(weak-head reduction),由此为纯 $\lambda$-演算定义出一种新的操作语义,且保持每个项原有的惰性(lazy)含义。一个只到达有限多个不同状态、且每个状态都在有限步内得解的项,会得到一张有限图,可能含环;演算保持纯性,所得图是忠实的(faithful),且与归约顺序无关。我们将这一操作语义实现为一个 $\lambda$-演算解释器。它自动完成动态规划,无需记忆化表即可共享重复的子问题;它创建并变换环图,无需任何额外的递归构造;它还能在有限时间内判定某些项没有弱头范式,并为 $\Omega$ 返回 $\bot$。

    求值器返回的是一张图,于是 $\lambda$-演算成为图计算的领域专用语言(DSL):动态规划的记忆表、博弈搜索的置换表(transposition table),以及 Datalog 可达性与指向分析(points-to analysis)的导出事实表,全都是在状态同一性上做 tabling,而它们无一需要手写。编译不过是又一个这样的问题:我们写出一个自举编译器(bootstrap compiler),它编译自身的源代码,而这一切都是一个纯 $\lambda$-项。
    \else
    Purely functional programming makes immutability and non-strict evaluation its defaults; we push further, making interning and memoization the defaults of computation. In prior operational semantics for the pure $\lambda$-calculus, representing and transforming cyclic and infinite data requires an added recursion construct, a $\texttt{letrec}$, a $\mu$-binder, or a built-in $Y$ for graph reduction, and sharing the repeated work of a memoized or dynamic-programming function requires an impure cache. We show that no extension is needed. We apply tabling, the standard method for solving a least-fixpoint equation, to weak-head reduction; this defines a new operational semantics for the pure $\lambda$-calculus that keeps each term's standard lazy meaning. A term that reaches finitely many distinct states, each solved in finitely many steps, comes out as a finite graph, possibly cyclic; the calculus stays pure, and the graph is faithful and independent of reduction order. We implemented this operational semantics as a $\lambda$-calculus interpreter. It does dynamic programming automatically, sharing repeated subproblems with no memoization table. It creates and transforms cyclic graphs with no added recursion construct. And it decides in finite time that some terms have no weak head normal form, returning $\bot$ for $\Omega$.

    What the evaluator returns is a graph, so the $\lambda$-calculus becomes a DSL for graph computation: the memo table of dynamic programming, the transposition table of game search, and the derived-fact table of Datalog reachability and points-to analysis are all tabling on state identity, and none of them is written by hand. Compilation is one more such problem: we write a bootstrap compiler that compiles its own source, all as a pure $\lambda$-term.
    \fi
    \ifTablambdaLNCS
    \keywords{\tablambdaKeywords{ \and}}
    \fi
  \end{abstract}

  \ifTablambdaLNCS\else
  \ccsdesc[500]{Theory of computation~Lambda calculus}
  \ccsdesc[500]{Theory of computation~Operational semantics}
  \ccsdesc[300]{Software and its engineering~General programming languages}
  \ccsdesc[300]{Theory of computation~Program analysis}

  \keywords{\tablambdaKeywords{,}}
  \fi
  \fi 

  \ifTablambdaLNCS\else
  \maketitle
  \fi

  \ifTablambdaBody
  \section{\bilingual{Introduction}{引言}}\label{sec:intro}

  \subsection{\bilingual{Manipulating Graphs in $\lambda$-Calculus}{在 $\lambda$-演算中操作图}}

  \ifPaperChinese
  纯 $\lambda$-演算中的普通归约从不构造环状数据结构。考虑零的流,即纯项 $r = Y\,(\mathtt{cons}\,0)$,其中 $Y$ 是不动点组合子,$\mathtt{cons}\,h\,t$ 表示在流 $t$ 前面接上 $h$。项 $r$ 归约为 $\mathtt{cons}\,0\,r$,于是它的尾部又是 $r$:这条流在自身上循环。然而 $r$ 没有有限的范式:普通归约把它无穷地展开成 $0,0,0,\dots$,从不把它折叠成一张有限的环形图。要得到一个有限的环,人们只能伸手到演算之外,取用一个额外的递归构造,比方说一个为该循环命名的 \texttt{letrec}~\cite{ariola1997-explicit-recursion}。

  造出环只是困难的一半。以教科书里的 $\mathtt{map}$ 为例,它把一个函数作用到列表的每个元素上,而它本身就是一个普通的纯 $\lambda$-项(附录~\ref{app:dsl})。在环状流 $r$ 上,项 $\mathtt{map}\;\mathtt{succ}\;r$(其中 $\mathtt{succ}$ 加一)归约为 $\mathtt{cons}\,1\,(\mathtt{map}\;\mathtt{succ}\;r)$,其内部的 $\mathtt{map}\;\mathtt{succ}\;r$ 又是同一个项,于是普通求值器会沿着这个环无穷地走下去,从不把结果折回成一个环。项中没有任何东西记录 $r$ 已被访问过。要让遍历停下来,函数自身就必须具备环感知能力:它必须记住已经访问过的节点,并识别出重新访问到的节点。$r$ 中的环不带任何纯项能够比较的名字,因此要识别它,要么需要一个按指针同一性比较节点的引用相等(reference equality)原语,而这是纯语言所禁止的,因为它能区分值相等的结构,从而破坏引用透明性(referential transparency),即值相等者可以自由替换这一规则;要么需要把图显式编码为节点带有名字的结构,比如在 $\lambda$-演算中模拟出来的内存的地址~\cite{ikuta2022-lambda-8cc},并在这些名字之上把教科书里的 $\mathtt{map}$ 改写成传递已访问名字集合的程序。

  因此,标准做法要付两次代价:一次扩展用来造环;而为了变换它,要么付出一次非纯,要么把程序改写成一个管理节点名字的程序。
  \else
  Ordinary reduction in the pure $\lambda$-calculus never builds a cyclic data structure. Consider the stream of zeros, the pure term $r = Y\,(\mathtt{cons}\,0)$, with the fixed-point combinator $Y$ and $\mathtt{cons}\,h\,t$ the stream $t$ with $h$ prepended. The term $r$ reduces to $\mathtt{cons}\,0\,r$, so its tail is $r$ again: the stream that loops on itself. Yet $r$ has no finite normal form: ordinary reduction unfolds it to $0,0,0,\dots$ forever and never folds it into a finite circular graph. To obtain a finite cycle one reaches outside the calculus for an added recursion construct, a \texttt{letrec}, say, that names the loop~\cite{ariola1997-explicit-recursion}.

  Creating the cycle is only half of the difficulty. Take the textbook $\mathtt{map}$, which applies a function to every element of a list and is itself an ordinary pure $\lambda$-term (Appendix~\ref{app:dsl}). On the cyclic stream $r$, the term $\mathtt{map}\;\mathtt{succ}\;r$ (with $\mathtt{succ}$ adding one) reduces to $\mathtt{cons}\,1\,(\mathtt{map}\;\mathtt{succ}\;r)$, whose inner $\mathtt{map}\;\mathtt{succ}\;r$ is the same term again, so an ordinary evaluator walks the cycle forever and never folds the result back into a cycle. Nothing in the term records that $r$ was already visited. To make the traversal terminate the function must itself become cycle-aware: it has to remember the nodes already visited and recognize a revisited node. The cycle in $r$ carries no name a pure term can compare, so recognizing it takes either a reference-equality primitive, which compares nodes by pointer identity and which a pure language forbids, since it tells apart structures of equal value and so breaks referential transparency, the rule that equal values may be substituted freely, or an explicit encoding of the graph whose nodes carry names, such as the addresses of a memory simulated in the $\lambda$-calculus~\cite{ikuta2022-lambda-8cc}, over which the textbook $\mathtt{map}$ must be rewritten to thread the set of names already visited.

  The standard account therefore pays twice: an extension to create the cycle, and to transform it either an impurity or a rewriting of the program into one that manages node names.
  \fi

  \subsection{\bilingual{Automating Dynamic Programming}{自动化动态规划}}\label{sec:intro-dp}

  \ifPaperChinese
  这一障碍并非无限数据所特有。编辑距离,即把一个字符串变成另一个所需的单字符插入、删除、替换的最少次数,也即 Levenshtein 距离~\cite{levenshtein1966-binary-codes},在两个字符串的后缀上有教科书式的递推~\cite{wagner1974-string-correction}:逐个头部比较,头部相同则在两个尾部上递归、不计代价,头部不同则计一次代价并取替换、删除、插入三者中的最优,而一旦某个字符串为空,代价就是另一个字符串剩余部分的长度。把它写成一个纯 $\lambda$-项,其中每个字符串是一个列表,对一个头/尾处理函数 $\lambda h.\lambda t$ 求值,为空时则取一个默认值,\footnote{这里的数据被编码为它自身的情形分析(Scott 编码与 Church 编码:前者是通常归于 Dana Scott 的未发表的 folklore,后者是 Church 的~\cite{church1941-calculi};二者在自解释中的运用见~\cite{mogensen1992-self-interpretation}):一个列表作用于两个参数时,非空则对其头与尾运行第一个参数,为空则返回第二个;一个布尔值作用于两个参数时,为真返回第一个、为假返回第二个,于是 $\mathtt{eq}\,h_a\,h_b$ 作用于其后两个项时便选出其一。}并以 $\mathtt{eq}$ 比较两个字符,$\mathtt{min}_3$ 取三者最小,$\mathtt{succ}$ 加一,$\mathtt{len}$ 取长度:
  \else
  The obstruction is not special to infinite data. Edit distance, the fewest single-character insertions, deletions, and substitutions that turn one string into another, the Levenshtein distance~\cite{levenshtein1966-binary-codes}, has the textbook recurrence~\cite{wagner1974-string-correction} on the strings' suffixes: reading head by head, equal heads recurse on both tails at no cost, unequal heads pay one and take the best of a substitution, a deletion, and an insertion, and once a string is empty the cost is the length of what remains of the other. Written as a pure $\lambda$-term, with each string a list that applies a head/tail handler $\lambda h.\lambda t$ or, when empty, a default,\footnote{Data here are encoded as their own case analysis (the Scott and Church encodings: the former unpublished folklore usually attributed to Dana Scott, the latter Church's~\cite{church1941-calculi}; see~\cite{mogensen1992-self-interpretation} for their use in self-interpretation): a list applied to two arguments runs the first on its head and tail or returns the second when empty, and a boolean applied to two arguments returns the first if true and the second if false, so $\mathtt{eq}\,h_a\,h_b$ applied to the two terms that follow it selects one of them.} and with $\mathtt{eq}$ comparing two characters, $\mathtt{min}_3$ the three-way minimum, $\mathtt{succ}$ adding one, and $\mathtt{len}$ the length:
  \fi
  \[
    \begin{aligned}
      &\mathtt{ed} = Y\,\big(\lambda\mathit{ed}.\,\lambda a.\,\lambda b.\\
        &\quad a\,\big(\lambda h_a.\lambda t_a.\\
          &\qquad b\,\big(\lambda h_b.\lambda t_b.\\
            &\qquad\quad (\mathtt{eq}\,h_a\,h_b)
            &&\text{\bilingual{compare the two heads}{比较两个头部}}\\
            &\qquad\qquad (\mathit{ed}\,t_a\,t_b)
            &&\text{\bilingual{equal: recurse on both tails, no cost}{相同:在两个尾部上递归,不计代价}}\\
            &\qquad\qquad \mathtt{succ}\,\big(\mathtt{min}_3\,
              \underbrace{(\mathit{ed}\,t_a\,t_b)}_{\text{\bilingual{substitute}{替换}}}\,
              \underbrace{(\mathit{ed}\,t_a\,b)}_{\text{\bilingual{delete}{删除}}}\\
            &\qquad\qquad\qquad
          \underbrace{(\mathit{ed}\,a\,t_b)}_{\text{\bilingual{insert}{插入}}}\big)\big)\\
        &\qquad (\mathtt{len}\,a)\big)
        &&\text{\bilingual{$b$ empty: delete all of $a$}{$b$ 为空:删除 $a$ 的全部}}\\
      &\quad (\mathtt{len}\,b)\big).
      &&\text{\bilingual{$a$ empty: insert all of $b$}{$a$ 为空:插入 $b$ 的全部}}
    \end{aligned}
  \]
  \ifPaperChinese
  这个项算出的距离是对的,但它的朴素求值是指数级的:三个递归调用反复展开相互重叠的后缀对。子问题是后缀对 $(t_a,t_b)$,而后缀是输入的共享子项、并非副本,因此若每个重复的对只计算一次,工作量就是长度为 $m$ 和 $n$ 的输入所对应的 $(m{+}1)(n{+}1)$ 个不同的对,也就是经典的 $O(mn)$ 动态规划表。但普通归约无从察觉某个对已经再次出现,于是重复计算。与环状的 $\mathtt{map}$ 一样,这个显然的纯程序正确却不可用:一个发散,另一个爆炸。
  \else
  This term computes the right distance, but its naive evaluation is exponential, the three recursive calls re-expanding overlapping suffix pairs. The subproblems are pairs of suffixes $(t_a,t_b)$, and a suffix is a shared sub-term of the inputs, not a copy, so if each repeated pair were computed once the work would be the $(m{+}1)(n{+}1)$ distinct pairs for inputs of lengths $m$ and $n$, the classic $O(mn)$ dynamic-programming table. But ordinary reduction has no way to notice that a pair has recurred, so it recomputes. Like the cyclic $\mathtt{map}$, the obvious pure program is correct yet unusable: one diverges, the other explodes.
  \fi

  \subsection{\bilingual{Tabling Weak-Head Reduction}{对弱头归约做 tabling}}

  \ifPaperChinese
  我们在不扩展演算的前提下解决这些问题。我们为纯 $\lambda$-演算给出一个解释器,以及一个编译器,在它们之下,正是这些程序得以高效运行而含义不变:$Y\,(\mathtt{cons}\,0)$ 求值为一张有限的环图,其上的 $\mathtt{map}\,\mathtt{succ}$ 求值为由一构成的有限的环,而 $\mathtt{ed}$ 以动态规划的 $O(mn)$ 次算术运算完成,每一个都是表~\ref{tab:regimes-normal-form} 与表~\ref{tab:regimes-no-normal-form} 所列出的某种情形的实例。本文附有一份实现\anon[,作为补充材料]{,公开可获取~\cite{yang2026-tablambda}}。我们运用 tabling:解释器把每个子项求值到它的弱头范式,把该范式的直接子项读作后继状态,把见过的状态制表(table),并在内化(interned)项的可判定的结构同一性上折叠一个递推,而该项所计算的那个函数上唯一的同一性,即行为相等,是不可判定的。这一识别是 tabling 所做的,而非项本身执行的操作,因此演算保持纯性。我们所计算的含义是已知的:Kurz 等人~\cite{kurz2013-nominal-coalgebraic-data-types} 用把一个项送到其弱头范式的那个余代数,以余代数的方式刻画了一个项的 Lévy--Longo 树;而 Milius 与 Wi{\ss}mann~\cite{milius2015-finitary-corecursion} 针对 Böhm 树指出,一个生成的子余代数在变量重命名的意义下有限的项,具有一棵有理树。那个结构映射是不可计算的,因为一个项是否有弱头范式是不可判定的,而两项工作都没有计算它;tabling 计算了它,在一个项所生成的子余代数有限时,把它构建为一张有限图(第~\ref{sec:bridge} 节)。
  \else
  We solve these problems without extending the calculus. We give an interpreter, and a compiler, for the pure $\lambda$-calculus under which these very programs run efficiently, their meaning unchanged: $Y\,(\mathtt{cons}\,0)$ evaluates to a finite cyclic graph, $\mathtt{map}\,\mathtt{succ}$ over it to the finite circle of ones, and $\mathtt{ed}$ in the dynamic program's $O(mn)$ arithmetic operations, each an instance of a regime Tables~\ref{tab:regimes-normal-form} and~\ref{tab:regimes-no-normal-form} set out. An implementation accompanies the paper \anon[as supplementary material]{openly available~\cite{yang2026-tablambda}}. We apply tabling: the interpreter evaluates each sub-term to its weak head normal form, reads the immediate sub-terms of that normal form as successor states, tables the states it has seen, and folds a recurrence on the decidable structural identity of the interned term, whereas the only identity on the function that term computes, behavioral equality, is undecidable. The recognition is tabling's, not an operation a term performs, so the calculus stays pure. The meaning we compute is known: Kurz et al.~\cite{kurz2013-nominal-coalgebraic-data-types} characterize the L\'evy--Longo tree of a term coalgebraically, by the coalgebra that sends a term to its weak head normal form, and Milius and Wi{\ss}mann~\cite{milius2015-finitary-corecursion} observe, for B\"ohm trees, that a term generating a subcoalgebra finite up to the renaming of variables has a rational tree. That structure map is not computable, since whether a term has a weak head normal form is undecidable, and neither work computes it; tabling does, building the subcoalgebra a term generates as a finite graph whenever it is finite (Section~\ref{sec:bridge}).
  \fi

  \newlength{\regimeConditionWidth}
  \newlength{\regimePreviousWidth}
  \newlength{\regimeOursWidth}
  \newlength{\regimeExampleWidth}
  \begin{table*}
    \footnotesize
    \settowidth{\regimeConditionWidth}{\bilingual{Finitely many, acyclic}{有限多个,无环}}
    \settowidth{\regimePreviousWidth}{\bilingual{Infinite tree}{无限树}}
    \settowidth{\regimeOursWidth}{\begin{tabular}{@{}l@{}}\bilingual{Finite cyclic graph}{有限环图}\\\bilingual{$\bot$ in finite time}{有限时间内得 $\bot$}\end{tabular}}
    \settowidth{\regimeExampleWidth}{$Y\,(\lambda s.\,\mathtt{cons}\,0\,(\mathtt{map}\,\mathtt{succ}\,s))$}
    \setlength{\tabcolsep}{5pt}
    \normalsize
    \caption{\bilingual{Terms all of whose states have a weak head normal form, by their states. \emph{Previous} is
        ordinary reduction, which unfolds the solution without tabling; \emph{Ours} is the tabled evaluation of
        Algorithm~\ref{alg:solvewhnf}. The two semantics realize the same denotation and differ only in
        representation, folding a rational infinite tree into a finite cyclic graph. In the second row every demand
        halts yet the graph is infinite, so its enumeration keeps producing forever.}{所有状态都有弱头范式的项,按其状态分类。
        \emph{以往}指普通归约,它不做 tabling 而直接展开解;\emph{本文}指算法~\ref{alg:solvewhnf} 的 tabled 求值。
        两种语义实现同一个指称,差别只在表示上:把一棵有理的无限树折叠成一张有限的环图。第二行里每一次求取都停机,
    但图是无限的,对它的枚举永远产出下去。}}
    \label{tab:regimes-normal-form}
    \centering\footnotesize
    \begin{tabular}{@{}p{\regimeConditionWidth}p{\regimePreviousWidth}p{\regimeOursWidth}p{\regimeExampleWidth}@{}}
      \hline
      \bilingual{States}{状态} & \bilingual{Previous}{以往} & \bilingual{Ours}{本文} & \bilingual{Example}{示例} \\
      \hline
      \bilingual{Finitely many, acyclic}{有限多个,无环} & \bilingual{Finite tree}{有限树} & \bilingual{Finite graph}{有限图} & $\mathtt{ed}\,a\,b$ \\
      \bilingual{Infinitely many}{无限多个} & \bilingual{Infinite tree}{无限树} & \bilingual{Infinite graph}{无限图} & $Y\,(\lambda s.\,\mathtt{cons}\,0\,(\mathtt{map}\,\mathtt{succ}\,s))$ \\
      \bilingual{Finitely many, cyclic}{有限多个,有环} & \bilingual{Infinite tree}{无限树} & \bilingual{Finite cyclic graph}{有限环图} & $Y\,(\mathtt{cons}\,0)$ \\
      \hline
    \end{tabular}

    \medskip
    \normalsize
    \caption{\bilingual{Terms with no weak head normal form, by the terms a demand visits; \emph{Previous} and
        \emph{Ours} are as in Table~\ref{tab:regimes-normal-form}. The two semantics realize the same denotation and
        differ only in termination, deciding the term as $\bot$ in finite time when its demand visits finitely many
        terms. The first row does not terminate, as the second row of Table~\ref{tab:regimes-normal-form} does not,
        but differs from it observably: here the first demand produces nothing.}{没有弱头范式的项,按一次求取所访问的项分类;
        \emph{以往}与\emph{本文}同表~\ref{tab:regimes-normal-form}。两种语义实现同一个指称,差别只在停机性上:
        当求取访问有限多个项时,在有限时间内把该项判定为 $\bot$。第一行与表~\ref{tab:regimes-normal-form} 的第二行一样不终止,
    但与之可观察地不同:这里第一次求取就毫无产出。}}
    \label{tab:regimes-no-normal-form}
    \centering\footnotesize
    \begin{tabular}{@{}p{\regimeConditionWidth}p{\regimePreviousWidth}p{\regimeOursWidth}p{\regimeExampleWidth}@{}}
      \hline
      \bilingual{Terms visited}{访问的项} & \bilingual{Previous}{以往} & \bilingual{Ours}{本文} & \bilingual{Example}{示例} \\
      \hline
      \bilingual{Infinitely many}{无限多个} & \bilingual{Diverges}{发散} & \bilingual{Diverges}{发散} & $Y\,(\lambda f.\lambda x.\,f\,(\mathtt{succ}\,x))\,0$ \\
      \bilingual{Finitely many}{有限多个} & \bilingual{Diverges}{发散} & \bilingual{$\bot$ in finite time}{有限时间内得 $\bot$} & $\Omega$ \\
      \hline
    \end{tabular}
  \end{table*}

  \subsection{\bilingual{Contributions}{贡献}}
  \ifPaperChinese
  \begin{description}
    \item[构造] 弱头归约的\emph{tabled 求值},它计算 Kurz 等人的余代数语义~\cite{kurz2013-nominal-coalgebraic-data-types}:一次一个项地把项求值到它的弱头范式,每当某个项
      再次出现便复用已制表的项,于是折叠出环的回边(back edge)不额外花费代价。一个没有弱头范式、
      且其求取访问有限多个项的项,会在有限时间内被判定为 $\bot$(未定义)
      (弱头方程见算法~\ref{alg:whnfmap},通用驱动器见算法~\ref{alg:whnf},二者组合成的求值器见算法~\ref{alg:solvewhnf})。
    \item[理论] 我们证明弱头方程的最小不动点就是大步弱头求值(引理~\ref{lem:mono-llt}),
      从而 tabling 的已知性质适用于它:一次求取在访问有限多个项时停机;求值器返回的图对 Lévy--Longo
      树是\emph{忠实的};并且这棵树只取决于各弱头范式,
      而不取决于它们如何被算出(第~\ref{sec:bridge} 节)。
    \item[应用] 一个普通的 $\mathtt{map}$ 把一个环形列表折叠成一个环形列表,而一次朴素的
      遍历会在求值发散之处停机(第~\ref{sec:bridge}、\ref{sec:application} 节);这一个解释器同时也是
      一门小型领域专用语言,求解 Datalog、图可达性与指向分析、极小化极大搜索,以及动态规划
      (附录~\ref{app:dsl});而编译本身就是一种图变换,同一个解释器把它求解成一个到 Python 的
      自举编译器(第~\ref{sec:graphs} 节)。
  \end{description}
  \else
  \begin{description}
    \item[The construction] \emph{Tabled evaluation} of weak-head reduction, which computes the coalgebraic semantics of Kurz et al.~\cite{kurz2013-nominal-coalgebraic-data-types}: evaluate
      one term at a time to its weak head normal form and reuse a tabled term whenever one recurs, so the back
      edges that fold a cycle cost nothing extra. A term with no weak head normal form whose
      demand visits finitely many terms is decided as $\bot$ (undefined) in finite time
      (the weak-head equation is Algorithm~\ref{alg:whnfmap}, the general driver Algorithm~\ref{alg:whnf}, and the evaluator they compose Algorithm~\ref{alg:solvewhnf}).
    \item[The theory] We prove that the least fixpoint of the weak-head equation is big-step
      weak-head evaluation (Lemma~\ref{lem:mono-llt}), so that the known properties of tabling apply to it: a
      demand halts when it visits finitely many terms, the graph the evaluator returns is
      \emph{faithful} to the L\'evy--Longo tree, and that tree depends only on the weak head normal forms, not on how they are computed
      (Section~\ref{sec:bridge}).
    \item[Applications] An ordinary
      $\mathtt{map}$ folds a circular list into a circular list and a naive walk terminates where
      evaluation diverges (Sections~\ref{sec:bridge}, \ref{sec:application}); the one interpreter is also a small
      domain-specific language solving Datalog, graph reachability and points-to analysis, minimax
      search, and dynamic programming (Appendix~\ref{app:dsl}); and compilation is itself a graph
      transformation, the same interpreter solving it into a bootstrap compiler to Python
      (Section~\ref{sec:graphs}).
  \end{description}
  \fi

  \section{\bilingual{A Tabled Operational Semantics}{一种 tabled 操作语义}}\label{sec:bridge}
  \ifPaperChinese
  本节讲的是如何把一个 $\lambda$-项的惰性指称语义求解成一张可能有环的图。
  \else
  This section is about how to solve a $\lambda$-term's lazy denotational semantics into a graph that may be cyclic.
  \fi

  \paragraph{\bilingual{The L\'evy--Longo tree}{Lévy--Longo 树}}
  \ifPaperChinese
  这套惰性语义就是该项的 Lévy--Longo 树~\cite{longo1983-set-theoretical-models, levy1978-reductions-lambda-calcul};我们从 Lévy 与 Longo 回顾它,并随之回顾它所基于的弱头归约。这里的项是 de Bruijn 形式的纯 $\lambda$-项:一个变量 $\Varsh_i$、一个抽象 $\Lamsh(t)$,或一个应用 $\Appsh(t, t')$,不含递归绑定子。\emph{弱头归约}反复收缩一个项的头部 redex,即沿其函数脊(function spine)最左的那个 redex:一个 redex $\Appsh(\Lamsh(b), a)$ 收缩为 $\textsf{Subst}(b, a)$,即把实参 $a$ 代入该抽象所绑定的那个变量的避免捕获(capture-avoiding)替换。当不再有头部 redex 时,该项处于\emph{弱头范式}(weak-head normal form),即一个变量、一个抽象,或一个以变量为首的应用,除非始终到不了范式而归约永不停止。一个项的 \emph{Lévy--Longo 树}~\cite{barendregt1984-lambda-calculus, wadsworth1971-semantics-lambda-calculus, abramsky1993-full-abstraction-lazy-lambda} 是它在这样一个唯一映射下的像:该映射从「把一个项送到其弱头范式、在它没有弱头范式时送到 $\bot$」的那个余代数通往终余代数~\cite{kurz2013-nominal-coalgebraic-data-types}。我们从 Courcelle~\cite{courcelle1983-infinite-trees} 回顾:这些树按「一棵树通过把 $\bot$ 叶替换为子树而得到细化」排序,构成一个带底的 $\omega$-完备偏序,其最小元是单个 $\bot$ 叶;而一棵树是\emph{有理的},当它只有有限多个不同子树,等价地,当它是一张有限图的展开。与之相对,一个项是一棵有限树,而把一棵有限树代入一棵有限树得到的仍是一棵有限树,因此弱头归约所到达的每个项都是有限的,而项上的结构同一性是可判定的。
  \else
  This lazy semantics is the term's L\'evy--Longo tree~\cite{longo1983-set-theoretical-models, levy1978-reductions-lambda-calcul}; we recall it from L\'evy and Longo, together with the weak-head reduction it is built on. The terms are pure $\lambda$-terms in de Bruijn form, a variable $\Varsh_i$, an abstraction $\Lamsh(t)$, or an application $\Appsh(t, t')$, with no recursion binder. \emph{Weak-head reduction} repeatedly contracts the head redex of a term, the leftmost redex along its function spine: a redex $\Appsh(\Lamsh(b), a)$ contracts to $\textsf{Subst}(b, a)$, the capture-avoiding substitution of the argument $a$ for the variable the abstraction binds. When no head redex remains the term stands in \emph{weak-head normal form}, a variable, an abstraction, or an application headed by a variable, unless no normal form is reached and reduction runs forever. The \emph{L\'evy--Longo tree} of a term~\cite{barendregt1984-lambda-calculus, wadsworth1971-semantics-lambda-calculus, abramsky1993-full-abstraction-lazy-lambda} is its image under the unique map into the final coalgebra from the coalgebra that sends a term to its weak head normal form, or to $\bot$ when it has none~\cite{kurz2013-nominal-coalgebraic-data-types}. We recall from Courcelle~\cite{courcelle1983-infinite-trees} that these trees, ordered so that a tree is refined by replacing $\bot$ leaves with subtrees, form a pointed $\omega$-complete partial order whose least element is the single $\bot$ leaf, and that a tree is \emph{rational} when it has finitely many distinct subtrees, equivalently when it is the unfolding of a finite graph. A term, by contrast, is a finite tree, and substituting a finite tree into a finite tree returns a finite tree, so every term weak-head reduction reaches is finite and structural identity on terms is decidable.
  \fi

  \paragraph{Tabling}
  \ifPaperChinese
  我们从 Tamaki 与 Sato~\cite{tamaki1986-tabled-resolution} 以及 Chen 与 Warren~\cite{chen1996-tabled-evaluation-delaying} 回顾 \emph{tabling},即从一个根出发求解最小不动点方程~\cite{vanemden1976-predicate-logic-semantics}的标准方法。这样的方程把每个状态映射到它的值,即覆于后继状态之上的一个顶层构造子,或映射到 $\bot$,而它的解把每个状态映射到这些值所展开成的那棵树。在逻辑程序中,方程由程序的子句给出,而 tabling 所计算的,是 van Emden 与 Kowalski 记作 $T_P$ 的那个算子的最小不动点~\cite{vanemden1976-predicate-logic-semantics}。Tabling 计算出该解:算出根的值,递归进入其后继状态,并保存一张已到达状态的表,以一个可判定的同一性为键;已在表中的状态被复用为同一个共享顶点,而仍在处理中的状态则成为一条闭合出环的回边。我们从 Tamaki 与 Sato 以及 Chen 与 Warren~\cite{tamaki1986-tabled-resolution, chen1996-tabled-evaluation-delaying} 回顾 tabling 的两条性质:当方程有最小不动点时,tabling 在一个被求取的状态处返回的值,就是最小不动点赋给它的那个值;并且只要从根出发到达的不同状态只有有限多个,tabling 就会停机,正如 SLG resolution 在具有有界项大小性质(bounded term-size property)的逻辑程序上那样~\cite{chen1996-tabled-evaluation-delaying}。这里没有任何东西是 $\lambda$-演算所特有的。
  \else
  We recall \emph{tabling} from Tamaki and Sato~\cite{tamaki1986-tabled-resolution} and from Chen and Warren~\cite{chen1996-tabled-evaluation-delaying}, the standard method for solving a least-fixpoint equation~\cite{vanemden1976-predicate-logic-semantics} from a root. Such an equation maps each state to its value, a top constructor over successor states, or to $\bot$, and its solution maps each state to the tree these values unfold to. In a logic program the equation is given by the program's clauses, and what tabling computes is the least fixpoint of the operator van Emden and Kowalski write $T_P$~\cite{vanemden1976-predicate-logic-semantics}. Tabling computes that solution: compute the root's value, recurse into its successor states, and keep a table of the states reached, keyed by a decidable identity; a state already tabled is reused as one shared vertex, and a state still being processed becomes a back edge that closes a cycle. We recall two properties of tabling from Tamaki and Sato and from Chen and Warren~\cite{tamaki1986-tabled-resolution, chen1996-tabled-evaluation-delaying}: when the equation has a least fixpoint, the value tabling returns at a demanded state is the one the least fixpoint assigns, and tabling halts whenever finitely many distinct states are reached from the root, as SLG resolution does on logic programs with the bounded term-size property~\cite{chen1996-tabled-evaluation-delaying}. Nothing here is specific to the $\lambda$-calculus.
  \fi

  \paragraph{\bilingual{Tabling weak-head reduction}{对弱头归约做 tabling}}
  \ifPaperChinese
  两者得以结合,是因为上面那个余代数就是这样一个方程,一个项的值就是它的弱头范式。算法~\ref{alg:whnfmap} 的 \textsf{WHNF} 就是这个方程:给定对被问及的任何项都返回其值的 \textsf{Resolve},\textsf{WHNF} 把变量或抽象原样返回,并解析应用的函数部分;当函数部分是抽象时,它收缩头部 redex 并返回收缩结果的值,这是一次尾调用,其值被直接返回;当函数部分是刚性的时,它返回覆于该函数部分与原实参之上的应用。把那个把 $g$ 送到 $t \mapsto \textsf{WHNF}(t, g)$ 的算子记作 $T_{\textsf{WHNF}}$,即 $T_P$ 的类比,我们所证明的唯一事实是它的最小不动点就是大步弱头求值,从而 Tamaki 与 Sato、Chen 与 Warren 所给出的 tabling 的性质适用于它。
  \else
  The two combine because the coalgebra above is such an equation, the value of a term being its weak head normal form. Algorithm~\ref{alg:whnfmap}, \textsf{WHNF}, is the equation: given \textsf{Resolve}, which returns the value of any term it is asked for, \textsf{WHNF} returns a variable or an abstraction as it is and resolves the function part of an application; when the function part is an abstraction it contracts the head redex and returns the value of the contractum, a tail call whose value is returned directly, and when the function part is rigid it returns the application over it and the original argument. Writing $T_{\textsf{WHNF}}$ for the operator that sends $g$ to $t \mapsto \textsf{WHNF}(t, g)$, the analogue of $T_P$, the one fact we prove is that its least fixpoint is big-step weak-head evaluation, so that Tamaki and Sato's and Chen and Warren's properties of tabling apply to it.
  \fi

  \begin{lemma}[\bilingual{Big-step weak-head evaluation is a least fixpoint}{大步弱头求值是一个最小不动点}]\label{lem:mono-llt}
    \ifPaperChinese
    按平坦序对从项到弱头范式或 $\bot$ 的函数逐点排序:在平坦序中 $\bot$ 位于每个弱头范式之下,而不同的弱头范式互不可比。$T_{\textsf{WHNF}}$ 在此序中的最小不动点是大步弱头求值,即把一个项送到其弱头范式、在它没有弱头范式时送到 $\bot$ 的那个函数。
    \else
    Order the functions from terms to weak head normal forms or $\bot$ pointwise by the flat order, in which $\bot$ lies below every weak head normal form and distinct weak head normal forms are incomparable. The least fixpoint of $T_{\textsf{WHNF}}$ in this order is big-step weak-head evaluation, the function that sends a term to its weak head normal form, or to $\bot$ when it has none.
    \fi
  \end{lemma}
  \begin{proof}
    \ifPaperChinese
    $T_{\textsf{WHNF}}$ 是单调且 $\omega$-连续的,因为 $\textsf{WHNF}(t, g)$ 至多在两个项处查询 $g$,即 $t$ 的函数部分与一个收缩结果;因此由 Kleene 不动点定理,它的最小不动点是它从 $\bot$ 出发的各次迭代的最小上界。把 \textsf{WHNF} 的三种情形读作推导 $t \Downarrow v$ 的规则;对推导与对头部归约的长度做归纳可知,$t \Downarrow v$ 可推导,恰好当从 $t$ 出发的头部归约到达弱头范式 $v$ 时。第 $k$ 次迭代在 $t \Downarrow v$ 有深度至多为 $k$ 的推导时把 $t$ 送到 $v$,否则送到 $\bot$;因此这些迭代的最小上界在 $t$ 有弱头范式时把它送到该弱头范式,而在它没有弱头范式时送到 $\bot$。
    \else
    $T_{\textsf{WHNF}}$ is monotone and $\omega$-continuous, since $\textsf{WHNF}(t, g)$ consults $g$ at no more than two terms, the function part of $t$ and one contractum, so by Kleene's fixed-point theorem its least fixpoint is the least upper bound of its iterates from $\bot$. Read the three cases of \textsf{WHNF} as rules deriving $t \Downarrow v$; by induction on derivations and on the length of head reductions, $t \Downarrow v$ is derivable exactly when the head reduction from $t$ reaches the weak head normal form $v$. The $k$-th iterate sends $t$ to $v$ when $t \Downarrow v$ has a derivation of depth at most $k$, and to $\bot$ otherwise, so their least upper bound sends $t$ to its weak head normal form when it has one and to $\bot$ when it has none.
    \fi
  \end{proof}

  \ifPaperChinese
  一次求取在计算一个值时所制表的状态,就是传给 \textsf{Resolve} 的那些项:被求取的项、它所归约的头脊上每个应用的函数部分,以及每个收缩结果;我们称它们为该求取所访问的项。由引理~\ref{lem:mono-llt} 与从 Tamaki 与 Sato、Chen 与 Warren 回顾的 tabling 的两条性质,$\textsf{TabledWHNF}(t)$ 返回 $t$ 的弱头范式,在它没有弱头范式时返回 $\bot$,并且只要一次求取访问的项是有限多个,它就会停机。
  \else
  The states tabled while a demand computes a value are the terms passed to \textsf{Resolve}: the demanded term, the function part of each application along a head spine it reduces, and each contractum; we call them the terms the demand visits. By Lemma~\ref{lem:mono-llt} and the two properties of tabling recalled from Tamaki and Sato and from Chen and Warren, $\textsf{TabledWHNF}(t)$ returns the weak head normal form of $t$, or $\bot$ when it has none, and a demand halts whenever it visits finitely many terms.
  \fi

  \begin{algorithm*}
    \DontPrintSemicolon
    \KwData{\bilingual{$\mathit{approx}$, each term's value, kept across rounds and across demands, $\bot$ until it is computed; $\mathit{cache}$, the values solved in the current round; $\mathit{stack}$, the terms whose value is being computed}{$\mathit{approx}$,每个项的值,跨轮且跨求取保留,在算出前为 $\bot$;$\mathit{cache}$,本轮已求解的值;$\mathit{stack}$,其值正在计算中的那些项}}
    \Fn{\TABLED{$\mathit{equation}$, $t$}}{
      \Fn{\RESOLVE{$u$}}{
        \lIf{$u \in \mathit{cache}$}{\Return $\mathit{cache}[u]$ \tcp*{\bilingual{solved this round}{本轮已求解}}}
        \lIf{$u \in \mathit{stack}$}{\Return $\mathit{approx}[u]$ \tcp*{\bilingual{re-entry: $\bot$ on the first round}{重入:首轮为 $\bot$}}}
        $\mathit{stack} \gets \mathit{stack} \cup \{u\}$\;
        $\mathit{value} \gets \mathit{equation}(u,\; \textsf{Resolve})$\;
        $\mathit{stack} \gets \mathit{stack} \setminus \{u\}$\;
        $\mathit{merged} \gets \mathit{approx}[u] \sqcup \mathit{value}$ \tcp*{\bilingual{deep merge in the approximation order; a conflict crashes}{在逼近序上深度合并;冲突则崩溃}}
        \lIf{$\mathit{merged} \neq \mathit{approx}[u]$}{$\mathit{changed} \gets \mathbf{true}$}
        $\mathit{approx}[u] \gets \mathit{merged}$;\; $\mathit{cache}[u] \gets \mathit{merged}$\;
        \Return $\mathit{merged}$\;
      }
      \Repeat{$\lnot\,\mathit{changed}$}{
        $\mathit{cache} \gets \emptyset$;\; $\mathit{stack} \gets \emptyset$;\; $\mathit{changed} \gets \mathbf{false}$\;
        \RESOLVE{$t$}\;
      }
      \Return $\mathit{approx}[t]$ \tcp*{\bilingual{the value of the demanded term}{被求取项的值}}
    }
    \caption{\bilingual{The general tabling driver. \textsf{Tabled} computes the value of the demanded term $t$ in the least fixpoint of an equation, by Kleene iteration from $\bot$: each round, through \textsf{Resolve}, the memoized resolution of the equation, it recomputes one value per term it visits and merges that value into the term's running approximation, a re-entry into a term still on the $\mathit{stack}$ returning that approximation, until a round adds nothing. The merge $\sqcup$ is the join in the approximation order, a deep merge that ascends rather than overwrites and crashes on a conflict, two incompatible non-$\bot$ values at one term; for the \textsf{WHNF} instance a value only ascends from $\bot$ to one weak head normal form, so the merge degenerates to that flat join. A recurring term is a $\mathit{cache}$ hit; the memo is keyed by the term itself, by structural identity, and $\mathit{approx}$ persists across demands, in the implementation a per-term cached property, so a term solved by an earlier demand is reused rather than re-solved. The equation is a parameter; \textsf{WHNF} (Algorithm~\ref{alg:whnfmap}) is the equation of weak-head reduction.}{通用 tabling 驱动器。\textsf{Tabled} 通过从 $\bot$ 出发的 Kleene 迭代计算被求取项 $t$ 在一个方程的最小不动点中的值:每一轮经由 \textsf{Resolve}(即该方程的记忆化求解)对每个访问到的项重算一个值,并把该值并入该项的运行中逼近,而对仍在 $\mathit{stack}$ 上的项的重入返回该逼近,直到某一轮不再增加任何东西。合并 $\sqcup$ 是逼近序上的并(join),一种向上攀升而非覆盖的深度合并,遇冲突即崩溃,即一个项处出现两个互不相容的非 $\bot$ 值;对 \textsf{WHNF} 这个实例,一个值只会从 $\bot$ 攀升到一个弱头范式,因此该合并退化为平坦的并。再次出现的项是一次 $\mathit{cache}$ 命中;备忘以项本身为键,即按结构同一性,且 $\mathit{approx}$ 跨求取保留(在实现中是每个项上的一个缓存属性),因此先前求取已解出的项被直接复用而非重解。方程是参数;\textsf{WHNF}(算法~\ref{alg:whnfmap})是弱头归约的方程。}}
    \label{alg:whnf}
  \end{algorithm*}

  \begin{algorithm*}
    \DontPrintSemicolon
    \Fn{\WHNF{$t$, \textsf{Resolve}}}{
      \Switch{$t$}{
        \uCase{$\Varsh\,i$}{\Return $\Varsh\,i$}
        \uCase{$\Lamsh\,b$}{\Return $\Lamsh\,b$ \tcp*{\bilingual{an abstraction is already a whnf; weak head stops at $\lambda$}{抽象本身已是 whnf;弱头在 $\lambda$ 处停止}}}
        \Case{$\Appsh\,f\,a$}{
          $h \gets \textsf{Resolve}(f)$\;
          \Switch{$h$}{
            \uCase{$\Lamsh\,b$}{\Return $\textsf{Resolve}(\Subst{b, a})$ \tcp*{\bilingual{head redex: $\beta$-contract and continue}{头部 redex:做 $\beta$-收缩并继续}}}
            \uCase{$\Varsh\,\_ \mid \Appsh\,\_$}{\Return $\Appsh\,h\,a$ \tcp*{\bilingual{rigid head: expose the application}{刚性头部:暴露该应用}}}
            \Case{$\bot$}{\Return $\bot$ \tcp*{\bilingual{the function has no whnf}{函数没有 whnf}}}
          }
        }
      }
    }
    \caption{\bilingual{The weak-head equation, passed to \textsf{Tabled} as $\mathit{equation}$; \textsf{WHNF} is the only function that reads the calculus.}{弱头方程,作为 $\mathit{equation}$ 传给 \textsf{Tabled};\textsf{WHNF} 是唯一读取演算的函数。}}
    \label{alg:whnfmap}
  \end{algorithm*}

  \begin{algorithm}
    \DontPrintSemicolon
    \Fn{\TABLEDWHNF{$t$}}{
      \Return \TABLED{\textsf{WHNF}, $t$}\;
    }
    \caption{\bilingual{The L\'evy--Longo evaluator: the weak-head equation \textsf{WHNF} tabled by the driver \textsf{Tabled}. $\textsf{TabledWHNF}(t)$ returns the weak head normal form of $t$, or $\bot$; the graph of $t$ is the assignment this induces, one weak head normal form per term of the subcoalgebra generated by $t$~\cite{rutten2000-universal-coalgebra}, each computed by a further demand.}{Lévy--Longo 求值器:弱头方程 \textsf{WHNF} 经驱动器 \textsf{Tabled} 做 tabling。$\textsf{TabledWHNF}(t)$ 返回 $t$ 的弱头范式,或 $\bot$;$t$ 的图就是由此诱导的赋值,即由 $t$ 生成的子余代数~\cite{rutten2000-universal-coalgebra} 中的每个项各一个弱头范式,每个由进一步的求取算出。}}
    \label{alg:solvewhnf}
  \end{algorithm}

  \paragraph{\bilingual{Faithfulness}{忠实性}}
  \ifPaperChinese
  该求值器对 Lévy--Longo 树是忠实的:从一个项 $t$ 出发对 \textsf{WHNF} 做 tabling,会算出一张展开为 $t$ 的 Lévy--Longo 树的图。$t$ 的图是 $t$ 所生成的子余代数~\cite{rutten2000-universal-coalgebra, milius2015-finitary-corecursion},因此按上文从 Kurz 等人~\cite{kurz2013-nominal-coalgebraic-data-types} 回顾的刻画,它展开为 $t$ 的 Lévy--Longo 树,并且由终性,这棵树只取决于各弱头范式,而不取决于它们如何被算出:一个在更粗的同一性上制表的编译运行时,只要算出相同的弱头范式,就得到同一棵树。其节点是内化项,各自携带进一步求取所算出的那个弱头范式,而再次出现的项按可判定的结构同一性就是同一个节点,\footnote{实现通过在项构造时对其做内化(intern,即 hash-consing),把这一同一性变成常数时间的指针测试;内化会终止,因为每个项都是一棵有限树;它是一项优化,加速 tabling 所需的同一性检查,而不影响 tabling 所计算出的结果。}这正是把一个递推折叠成共享节点或回边的机制。枚举或画出这张图,如第~\ref{sec:application} 节的 trace 与配图以及随附测试所做的那样,是元语言中的观察:逐节点求取一次,并按节点同一性检测环;演算自身做不出这种观察,因为检测环恰是它无法读回的那种共享(附录~\ref{app:compiler})。
  \else
  The evaluator is faithful to the L\'evy--Longo tree: tabling \textsf{WHNF} from a term $t$ computes a graph that unfolds to the L\'evy--Longo tree of $t$. The graph of $t$ is the subcoalgebra $t$ generates~\cite{rutten2000-universal-coalgebra, milius2015-finitary-corecursion}, so it unfolds to the L\'evy--Longo tree of $t$ by the characterization of Kurz et al.~\cite{kurz2013-nominal-coalgebraic-data-types} recalled above, and by finality that tree depends only on the weak head normal forms, not on how they are computed: a compiled runtime that tables on a coarser identity yields the same tree wherever it computes the same weak head normal forms. Its nodes are interned terms, each carrying the weak head normal form a further demand computes, and a recurring term is the same node by the decidable structural identity,\footnote{The implementation makes this identity a constant-time pointer test by interning (hash-consing) terms as they are built; interning terminates because every term is a finite tree, and it is an optimization that speeds the identity check tabling needs without affecting what tabling computes.} which is what folds a recurrence into a shared node or a back edge. Enumerating or drawing that graph, as the traces and figures of Section~\ref{sec:application} and the accompanying tests do, is observation in the metalanguage, one demand per node with cycle detection on node identity; the calculus itself cannot make that observation, since detecting a cycle is exactly the sharing it cannot read back (Appendix~\ref{app:compiler}).
  \fi

  \section{\bilingual{Case Study: Three Regimes}{案例研究:三种情形}}\label{sec:application}

  \ifPaperChinese
  我们用 Claude Code 把算法~\ref{alg:solvewhnf} 的 tabled 求值实现为一个纯 $\lambda$-演算的解释器,\anon[作为补充材料提供]{可获取~\cite{yang2026-tablambda}},并在表~\ref{tab:regimes-normal-form} 与表~\ref{tab:regimes-no-normal-form} 所列情形中的三种上运行它,即解释器与普通归约分道扬镳的那三种。
  \else
  We implemented the tabled evaluation of Algorithm~\ref{alg:solvewhnf} as an interpreter for the pure $\lambda$-calculus with Claude Code, available\anon[ as supplementary material]{~\cite{yang2026-tablambda}}, and run it on three of the regimes Tables~\ref{tab:regimes-normal-form} and~\ref{tab:regimes-no-normal-form} set out, the three in which the interpreter parts from ordinary reduction.
  \fi

  \subsection{\bilingual{Sharing Subproblems: Edit Distance}{共享子问题:编辑距离}}\label{sec:application-sharing}

  \ifPaperChinese
  这是表~\ref{tab:regimes-normal-form} 的第一种情形:每个状态都有弱头范式,有限多个状态构成一棵有限的调用树,被解释器共享成一张有限图。第~\ref{sec:intro-dp} 节的项 $\mathtt{ed}$ 就是把教科书式的递推写成纯 $\lambda$-演算;附录~\ref{app:editdistance-code} 连同其库一并完整列出它,直到 $Y$ 与 BinNat 算术。那个驯服其指数级朴素求值的唯一观察,即相同的子问题只需计算一次,恰恰正是解释器所做的,而且是免费做到的:$\mathtt{ed}$ 中任何地方都不出现记忆化表、不出现已访问集合,也不出现引用相等测试。本节追踪它是如何做到的。
  \else
  This is the first regime of Table~\ref{tab:regimes-normal-form}: every state has a weak head normal form, and the finitely many states form a finite call tree the interpreter shares into a finite graph. The term $\mathtt{ed}$ of Section~\ref{sec:intro-dp} is the textbook recurrence written in the pure $\lambda$-calculus; Appendix~\ref{app:editdistance-code} lists it in full with its library, down to $Y$ and the BinNat arithmetic. The one observation that tames its exponential naive evaluation, that identical subproblems need be computed only once, is exactly what the interpreter makes, and makes for free: no memoization table, no visited set, and no reference-equality test appears anywhere in $\mathtt{ed}$. This section traces how.
  \fi

  \paragraph{\bilingual{Subproblems are shared sub-terms}{子问题是共享的子项}}
  \ifPaperChinese
  一个子问题就是一次调用 $\mathit{ed}\,t_a\,t_b$,而它的两个参数是输入的后缀。Scott 编码列表的后缀并不是新建的列表,而是原列表的一个子项:尾处理函数 $\lambda h_a.\lambda t_a$ 把 $t_a$ 绑定到那个已经代表 $a$ 之剩余部分的节点本身。因此一个给定的后缀对,在每一次出现时,都是同一对节点,而由它们组装出的应用 $\mathit{ed}\,t_a\,t_b$ 是同一个内化节点。\textsf{Tabled} 以第~\ref{sec:bridge} 节确立的内化项的常数时间结构同一性为表的键,因此一个后缀对第二次出现时,它的调用已是一个制了表的状态,被直接读出而非重新进入。
  \else
  A subproblem is a call $\mathit{ed}\,t_a\,t_b$, and its two arguments are suffixes of the inputs. A suffix of a Scott-encoded list is not a freshly built list but a sub-term of the original: the tail handler $\lambda h_a.\lambda t_a$ binds $t_a$ to the very node that already represents the rest of $a$. So a given suffix pair is, on every occurrence, the same pair of nodes, and the application $\mathit{ed}\,t_a\,t_b$ assembled from them is one and the same interned node. \textsf{Tabled} keys its table by the constant-time structural identity of interned terms established in Section~\ref{sec:bridge}, so the second time a suffix pair arises its call is already a tabled state, read off rather than re-entered.
  \fi

  \paragraph{\bilingual{The collapse, on a tiny instance}{在一个小例子上的坍缩}}
  \ifPaperChinese
  取 $a=\mathtt{ab}$ 与 $b=\mathtt{cd}$(图~\ref{fig:editdistance})。头部不同,因此 $\mathit{ed}\,\mathtt{ab}\,\mathtt{cd}$ 等于一加上三个调用中的最小者:替换对应 $\mathit{ed}\,\mathtt{b}\,\mathtt{d}$,删除对应 $\mathit{ed}\,\mathtt{b}\,\mathtt{cd}$,插入对应 $\mathit{ed}\,\mathtt{ab}\,\mathtt{d}$。每个内部调用又各自分出三支,而同一个子问题会沿其中几条路径被到达:例如 $\mathit{ed}\,\mathtt{b}\,\mathtt{d}$ 就被 $\mathit{ed}\,\mathtt{ab}\,\mathtt{cd}$、$\mathit{ed}\,\mathtt{ab}\,\mathtt{d}$、$\mathit{ed}\,\mathtt{b}\,\mathtt{cd}$ 三者全都调用。朴素递归会重新展开每一次这样的重复,这里一共 $19$ 次调用。解释器不会:第一个参数取遍 $a$ 的三个后缀 $\{\mathtt{ab},\mathtt{b},\varepsilon\}$,第二个取遍 $b$ 的三个后缀 $\{\mathtt{cd},\mathtt{d},\varepsilon\}$,因此只有 $3\times 3=9$ 个不同的调用。沿几条路径被到达的调用,是一个带几条入边、只求值一次的制表节点,即图中高亮的状态;这九个节点承载着 $a$ 的某后缀与 $b$ 的某后缀之间的距离,这正是我们熟悉的编辑距离表。这种坍缩并非由 $\mathtt{ed}$ 安排;它是算法~\ref{alg:solvewhnf} 的 tabled 求值在运行时所做的。第一个调用 $\mathit{ed}\,\mathtt{b}\,\mathtt{d}$ 尚不在表中,于是解释器求解它并记录其弱头范式。当删除分支稍后请求同一个项时,它是同一个内化节点、已在表中,其弱头范式被直接返回而无需再次求解:一次 cache 命中,即图中入到某一节点的第二条边。附录~\ref{app:editdistance-trace} 是整个运行过程,每个子问题的求解与之后每一次出现的命中,按解释器到达它们的次序排列;它由对解释器插桩生成,附录~\ref{app:editdistance-code} 的列表由源项生成,均非手工誊写。
  \else
  Take $a=\mathtt{ab}$ and $b=\mathtt{cd}$ (Figure~\ref{fig:editdistance}). The heads differ, so $\mathit{ed}\,\mathtt{ab}\,\mathtt{cd}$ is one plus the least of three calls: $\mathit{ed}\,\mathtt{b}\,\mathtt{d}$ for a substitution, $\mathit{ed}\,\mathtt{b}\,\mathtt{cd}$ for a deletion, and $\mathit{ed}\,\mathtt{ab}\,\mathtt{d}$ for an insertion. Each interior call branches into three again, and the same subproblem is reached along several of these paths: $\mathit{ed}\,\mathtt{b}\,\mathtt{d}$, for instance, is called by all three of $\mathit{ed}\,\mathtt{ab}\,\mathtt{cd}$, $\mathit{ed}\,\mathtt{ab}\,\mathtt{d}$, and $\mathit{ed}\,\mathtt{b}\,\mathtt{cd}$. The naive recursion re-expands every such repeat, $19$ calls in all here. The interpreter does not: the first argument ranges over the three suffixes $\{\mathtt{ab},\mathtt{b},\varepsilon\}$ of $a$ and the second over the three suffixes $\{\mathtt{cd},\mathtt{d},\varepsilon\}$ of $b$, so there are only $3\times 3=9$ distinct calls. A call reached along several paths is one tabled node with several incoming edges, evaluated once, the highlighted states in the figure; the nine nodes carry the distances between a suffix of $a$ and a suffix of $b$, which is the familiar edit-distance table. The collapse is not arranged by $\mathtt{ed}$; it is what the tabled evaluation of Algorithm~\ref{alg:solvewhnf} does as it runs. The first call, $\mathit{ed}\,\mathtt{b}\,\mathtt{d}$, is not yet in the table, so the interpreter solves it and records its weak head normal form. When the deletion branch later asks for that same term, it is the same interned node, already in the table, and its weak head normal form is returned without solving it again: a cache hit, the second edge into one node in the figure. Appendix~\ref{app:editdistance-trace} is the whole run, every subproblem solved and every later occurrence hit, in the order the interpreter reaches them; it is generated by instrumenting the interpreter, and the listing of Appendix~\ref{app:editdistance-code} from the source terms, neither transcribed by hand.
  \fi

  \begin{figure}
    \centering
\begin{tikzpicture}[
  font=\footnotesize,
  state/.style={draw, rounded corners, minimum width=10mm, minimum height=8mm, align=center, inner sep=1pt},
  shared/.style={draw=red!65, fill=red!10, rounded corners, minimum width=10mm, minimum height=8mm, align=center, inner sep=1pt},
  axis/.style={font=\footnotesize\itshape},
  call/.style={-Stealth, gray!55, shorten >=1pt, shorten <=1pt},
]
  \node[axis] at (0.00,1.01) {$\mathtt{cd}$};
  \node[axis] at (1.75,1.01) {$\mathtt{d}$};
  \node[axis] at (3.50,1.01) {$\varepsilon$};
  \node[axis] at (-1.43,0.00) {$\mathtt{ab}$};
  \node[axis] at (-1.43,-1.30) {$\mathtt{b}$};
  \node[axis] at (-1.43,-2.60) {$\varepsilon$};
  \node[state] (s0x0) at (0.00,0.00) {$2$};
  \node[state] (s0x1) at (1.75,0.00) {$2$};
  \node[state] (s0x2) at (3.50,0.00) {$2$};
  \node[state] (s1x0) at (0.00,-1.30) {$2$};
  \node[shared] (s1x1) at (1.75,-1.30) {$1$};
  \node[shared] (s1x2) at (3.50,-1.30) {$1$};
  \node[state] (s2x0) at (0.00,-2.60) {$2$};
  \node[shared] (s2x1) at (1.75,-2.60) {$1$};
  \node[state] (s2x2) at (3.50,-2.60) {$0$};
  \draw[call] (s0x0) -- (s1x1);
  \draw[call] (s0x0) -- (s1x0);
  \draw[call] (s0x0) -- (s0x1);
  \draw[call] (s0x1) -- (s1x2);
  \draw[call] (s0x1) -- (s1x1);
  \draw[call] (s0x1) -- (s0x2);
  \draw[call] (s1x0) -- (s2x1);
  \draw[call] (s1x0) -- (s2x0);
  \draw[call] (s1x0) -- (s1x1);
  \draw[call] (s1x1) -- (s2x2);
  \draw[call] (s1x1) -- (s2x1);
  \draw[call] (s1x1) -- (s1x2);
  \node[anchor=north, align=center] at (1.75,-3.71) {\footnotesize naive recursion: 19 calls $\;\Longrightarrow\;$ tabled: 9 states};
\end{tikzpicture}
    \caption{\bilingual{The subproblem call graph of $\mathtt{ed}$ on $a=\mathtt{ab}$, $b=\mathtt{cd}$, generated by running
        the interpreter. Each node is one of the $(m{+}1)(n{+}1)=9$ distinct suffix-pair states, labeled with
        its edit distance; rows are suffixes of $a$, columns suffixes of $b$. An interior state calls the three
        subproblems it points to. The highlighted nodes have more than one incoming edge; the naive recursion
      would make $19$ calls.}{$\mathtt{ed}$ 在 $a=\mathtt{ab}$、$b=\mathtt{cd}$
        上的子问题调用图,由运行解释器生成。每个节点是 $(m{+}1)(n{+}1)=9$ 个不同的后缀对状态之一,标注其编辑距离;
        行是 $a$ 的后缀,列是 $b$ 的后缀。一个内部状态调用它所指向的三个子问题。高亮节点带不止一条入边;朴素递归
    则会做 $19$ 次调用。}}
    \label{fig:editdistance}
  \end{figure}

  \paragraph{\bilingual{The cost}{代价}}
  \ifPaperChinese
  对于长度为 $m$ 和 $n$ 的输入,后缀对是 $a$ 的 $m{+}1$ 个后缀与 $b$ 的 $n{+}1$ 个后缀的 $(m{+}1)(n{+}1)$ 个乘积。解释器把这些状态各计算一次,即 trace 中的 \texttt{compute} 行,并把每次重复的调用折叠为一条回边,即 \texttt{hit} 行。每个状态做的算术运算次数有界:它的两个头部作为固定字母表上的定宽二进制码比较,其代价是对二进制自然数的一次后继与一次三路最小。以这些运算为计数单位(经典账目正是以此陈述其界),这 $(m{+}1)(n{+}1)$ 个状态在 $O(mn)$ 的时间与空间内完成动态规划表的工作;在项的层面,一个距离至多为 $m+n$,一次后继或最小至多走过其操作数的 $\log(m+n)$ 个二进制位,因此求解器花费 $O(mn\log(m+n))$ 的步数与制表条目,这个对数因子是算术本身也是纯项的代价。$\mathtt{ed}$ 算出的确实是真正的 Levenshtein 距离,这一点由一个通过的测试再现,该测试在一系列输入上把它与教科书式的递推交叉核对\anon[~(补充材料)]{~\cite{yang2026-tablambda}}。
  \else
  For inputs of lengths $m$ and $n$ the suffix pairs are the $(m{+}1)(n{+}1)$ products of the $m{+}1$ suffixes of $a$ with the $n{+}1$ suffixes of $b$. The interpreter computes each of these states once, the \texttt{compute} lines of the trace, and folds every repeated call into a back edge, the \texttt{hit} lines. Each state does a bounded number of arithmetic operations: its two heads are compared as fixed-width binary codes over a fixed alphabet, and its cost is a single successor and one three-way minimum on binary naturals. Counted in these operations, the unit in which the classical account states its bound, the $(m{+}1)(n{+}1)$ states do the work of the dynamic-programming table in $O(mn)$ time and space; at the term level a distance is at most $m+n$, a successor or minimum walks at most the $\log(m+n)$ bits of its operands, and the solver spends $O(mn\log(m+n))$ steps and table entries, the logarithmic factor the price of the arithmetic itself being a pure term. That $\mathtt{ed}$ computes the genuine Levenshtein distance is reproduced as a passing test that cross-checks it against the textbook recurrence over a range of inputs\anon[~(supplementary material)]{~\cite{yang2026-tablambda}}.
  \fi

  \subsection{\bilingual{Folding a Cycle: The Stream of Zeros}{折叠一个环:零的流}}\label{sec:application-cycles}

  \ifPaperChinese
  零的流 $r = Y\,(\mathtt{cons}\,0)$ 是表~\ref{tab:regimes-normal-form} 的第三种情形:每个状态都有弱头范式,而有限多个状态的无限树被解释器折叠成一张有限的环图。它由普通的函数调用构成,项中并未写入任何环;附录~\ref{app:cyclic-zeros-code} 连同其库一并完整列出它。我们跟随解释器计算它,看着这个环成形。
  \else
  The stream of zeros $r = Y\,(\mathtt{cons}\,0)$ is the third regime of Table~\ref{tab:regimes-normal-form}: every state has a weak head normal form, and the finitely many states have an infinite tree the interpreter folds into a finite cyclic graph. It is built from ordinary function calls, with no cycle written into the term; Appendix~\ref{app:cyclic-zeros-code} lists it in full with its library. We follow the interpreter computing it, and watch the cycle form.
  \fi

  \paragraph{\bilingual{What the tail is}{尾部是什么}}
  \ifPaperChinese
  一个 Scott cons 单元 $\lambda c.\lambda n.\,c\,h\,t$ 是一个延续(continuation):作用于一个 cons 处理函数 $c$ 和一个 nil 处理函数 $n$ 时,它对其头 $h$ 与尾 $t$ 调用 $c$,因此尾部就是该单元交给 $c$ 的那个项。对 $r$ 而言,这个项不是「又一个 $r$」,而是不动点组合子产生的自应用 $W\,W$,其中 $W = \lambda x.\,\mathtt{cons}\,0\,(x\,x)$:把 $Y$ 展开一次会把 $r$ 改写成 $W\,W$,而 $W\,W$ 暴露出单元 $\mathtt{cons}\,0\,(W\,W)$,其尾部是 $W\,W$。
  \else
  A Scott cons cell $\lambda c.\lambda n.\,c\,h\,t$ is a continuation: applied to a cons handler $c$ and a nil handler $n$, it calls $c$ on its head $h$ and tail $t$, so the tail is whatever term the cell hands to $c$. For $r$ that term is not ``$r$ again'' but the self-application $W\,W$ that the fixed-point combinator produces, with $W = \lambda x.\,\mathtt{cons}\,0\,(x\,x)$: unfolding $Y$ once rewrites $r$ to $W\,W$, and $W\,W$ exposes the cell $\mathtt{cons}\,0\,(W\,W)$, whose tail is $W\,W$.
  \fi

  \paragraph{\bilingual{How the interpreter folds it}{解释器如何折叠它}}
  \ifPaperChinese
  \textsf{Tabled} 以它所求值的内化项为表的键(算法~\ref{alg:whnf})。求取 $r$ 会运行 \textsf{WHNF}(算法~\ref{alg:whnfmap}),它收缩头部 redex 直到 cons 单元被暴露;该单元内的尾部是 $W\,W$,而求取 $W\,W$ 暴露出同一个单元,其尾部又一次是 $W\,W$。这次折叠系于一个事实:第二个 $W\,W$ 是由替换从头重建的,是一个全新的项,call-by-name 或 call-by-need 求值器会不停地重建它,无尽地流出 $0,0,0,\dots$。解释器按结构内化项,因此重建出的 $W\,W$ 正是其弱头范式刚被求取的那个节点,已经是图上的一个节点,于是这个尾部是指回它的一条回边,而不是又一个待求取的项。$r$ 的图就是一个尾部指回自身的 cons 单元(图~\ref{fig:cyclic-zeros}),即展开为 $0,0,0,\dots$ 的那张有限环图。整个运行过程见附录~\ref{app:cyclic-zeros-trace},由运行解释器生成。
  \else
  \textsf{Tabled} keys its table by the interned term it evaluates (Algorithm~\ref{alg:whnf}). Demanding $r$ runs \textsf{WHNF} (Algorithm~\ref{alg:whnfmap}), which contracts head redexes until the cons cell is exposed; the tail inside that cell is $W\,W$, and demanding $W\,W$ exposes the same cell, whose tail is $W\,W$ once more. The fold turns on one fact: that second $W\,W$ is rebuilt from scratch by the substitution, a fresh term a call-by-name or call-by-need evaluator would keep rebuilding, streaming $0,0,0,\dots$ without end. The interpreter interns terms by their structure, so the rebuilt $W\,W$ is the very node whose weak head normal form was just demanded, already a node of the graph, and that tail is a back edge to it rather than a further term to demand. The graph of $r$ is one cons cell whose tail points back to it (Figure~\ref{fig:cyclic-zeros}), the finite cyclic graph that unfolds to $0,0,0,\dots$. The whole run is Appendix~\ref{app:cyclic-zeros-trace}, generated by running the interpreter.
  \fi

  \paragraph{\bilingual{The same recognition as edit distance}{与编辑距离相同的识别}}
  \ifPaperChinese
  这正是第~\ref{sec:application-sharing} 节中共享子问题的那种识别,只是读法上差了一档。在那里,重复的调用是一个已经求解完的调用,形成一个菱形,几个调用方在一个制表节点处汇合;而在这里,调用再次出现在它自己所算出的那个弱头范式之内,因此复用它就是一条回边,一个环。两种情形里,项本身都不携带任何共享或环:每一步都是函数调用,而那张图,无论是菱形还是圆环,都是 tabling 对哪些调用再次出现的记录。
  \else
  This is the recognition that shared a subproblem in Section~\ref{sec:application-sharing}, read one notch differently. There a repeated call was one already solved, a diamond where several callers meet at one tabled node; here the call recurs inside the very weak head normal form it computes, so the reuse is a back edge, a cycle. In neither does the term carry sharing or a cycle of its own: every step is a function call, and the graph, diamond or circle, is tabling's record of which calls recur.
  \fi

  \begin{figure}
    \centering
\begin{tikzpicture}[
  font=\footnotesize,
  state/.style={draw, rounded corners, minimum width=9mm, minimum height=8mm, align=center, inner sep=1pt},
  shared/.style={draw=red!65, fill=red!10, rounded corners, minimum width=9mm, minimum height=8mm, align=center, inner sep=1pt},
  call/.style={-Stealth, gray!60, shorten >=1pt, shorten <=1pt},
  back/.style={-Stealth, red!65, shorten >=1pt, shorten <=1pt},
]
  \node[state] (s0) at (0.00,0) {$0$};
  \node[font=\footnotesize\itshape, below=1.5pt of s0] {$r$};
  \node[shared] (s1) at (2.80,0) {$0$};
  \node[font=\footnotesize\itshape, below=1.5pt of s1] {$W\,W$};
  \draw[call] (s0) -- node[above]{\itshape tail} (s1);
  \draw[back] (s1) to[out=40,in=-40,looseness=8] node[right]{\itshape tail} (s1);
  \node[anchor=north, align=center] at (1.40,-1.2) {\footnotesize each state exposes $\mathtt{cons}\;0\;\cdot$; unfolds to $0,0,0,\ldots$};
\end{tikzpicture}
    \caption{\bilingual{The graph of $r = Y\,(\mathtt{cons}\,0)$, generated by running the interpreter. Its nodes are $r$ and the
        self-application $W\,W$ (with $W = \lambda x.\,\mathtt{cons}\,0\,(x\,x)$), each exposing the cons cell with head $0$;
        the demands also table the head-chain terms they pass through, such as $\mathtt{cons}\,0\,(W\,W)$, memo entries
        rather than nodes of the graph, and not drawn. The drawing is data-level: a cell's head and tail are read off
        the exposed abstraction's body, compressing the abstraction and application nodes of the cell into one
        drawn state. The highlighted self-loop is the tail of $W\,W$, which is $W\,W$ again.}{$r = Y\,(\mathtt{cons}\,0)$ 的图,由运行解释器生成。
        其节点是 $r$ 与自应用 $W\,W$(其中 $W = \lambda x.\,\mathtt{cons}\,0\,(x\,x)$),各自暴露出头为 $0$ 的 cons 单元;
        各次求取途中经过的头链项(如 $\mathtt{cons}\,0\,(W\,W)$)也会被制表,但那是备忘条目而非图的节点,故不画出。
        画法取数据层面:单元的头与尾从暴露的抽象体上读出,单元自身的抽象与应用节点被压缩为一个画出的状态。
    高亮的自环是 $W\,W$ 的尾部,它又是 $W\,W$。}}
    \label{fig:cyclic-zeros}
  \end{figure}

  \subsection{\bilingual{Deciding a Term with No Weak Head Normal Form and Finitely Many Terms Visited: $\Omega$}{判定一个没有弱头范式、访问有限多个项的项:$\Omega$}}\label{sec:application-bottom}

  \ifPaperChinese
  项 $\Omega = (\lambda x.\,x\,x)\,(\lambda x.\,x\,x)$ 是表~\ref{tab:regimes-no-normal-form} 的第二种情形,没有弱头范式、且访问有限多个项,在此普通归约发散,而解释器在有限时间内返回 $\bot$。它是 $\omega = \lambda x.\,x\,x$ 的自应用 $\omega\,\omega$,列于附录~\ref{app:omega-code},没有任何构造子可暴露。它唯一的头部 $\beta$-收缩把 $x\,x$ 中的 $x$ 替换为 $\lambda x.\,x\,x$,又得到 $\Omega$:头部归约回到它出发之处,却从未到达弱头范式。
  \else
  The term $\Omega = (\lambda x.\,x\,x)\,(\lambda x.\,x\,x)$ is the second regime of Table~\ref{tab:regimes-no-normal-form}, with no weak head normal form and finitely many terms visited, where ordinary reduction diverges and the interpreter returns $\bot$ in finite time. It is the self-application $\omega\,\omega$ of $\omega = \lambda x.\,x\,x$, listed in Appendix~\ref{app:omega-code}, and has no constructor to expose. Its one head $\beta$-contraction substitutes $\lambda x.\,x\,x$ for $x$ in $x\,x$, giving $\Omega$ again: the head reduction returns to where it started without ever reaching a weak head normal form.
  \fi

  \paragraph{\bilingual{Stepping through the interpreter}{逐步走过解释器}}
  \ifPaperChinese
  求解 $\Omega$(算法~\ref{alg:whnf})会把它压入栈并运行 \textsf{WHNF}(算法~\ref{alg:whnfmap}),后者收缩头部 redex;因为项是内化的,收缩结果并不是一份全新的副本,而是与 $\Omega$ 同一个节点,即仍在栈上的那个项,且什么都没暴露,于是这次重入返回它的运行中逼近 $\bot$。重新计算仍得 $\bot$,迭代已收敛,解释器在普通归约会无尽收缩同一 redex 之处停机(图~\ref{fig:omega})。第~\ref{sec:application-cycles} 节同样遇到了一个再次出现的项,只是它出现在弱头范式的另一侧:在那里,$W\,W$ 再次出现在它所暴露的 cons 单元之内,是留待日后求取的尾部,即图中的那条回边;而在这里,$\Omega$ 再次出现在尾调用链自身之中,尚未到达任何弱头范式,于是重入遇到的逼近仍是 $\bot$,该项没有弱头范式。判定它,正是第~\ref{sec:bridge} 节那个从 $\bot$ 起的攀升发挥作用之处。附录~\ref{app:omega-trace} 是该运行过程,由解释器生成。
  \else
  Solving $\Omega$ (Algorithm~\ref{alg:whnf}) puts it on the stack and runs \textsf{WHNF} (Algorithm~\ref{alg:whnfmap}), which contracts the head redex; because terms are interned, the contractum is not a fresh copy but the same node as $\Omega$, the term still on the stack, and nothing has been exposed, so the re-entry returns its running approximation, $\bot$. Recomputing leaves it $\bot$, the iteration has converged, and the interpreter halts where ordinary reduction would contract the same redex without end (Figure~\ref{fig:omega}). Section~\ref{sec:application-cycles} also met a recurring term, but on the other side of the weak head normal form: there $W\,W$ recurs inside the cons cell it exposes, a tail for a later demand, the back edge of the graph; here $\Omega$ recurs within the tail-call chain itself, before any weak head normal form is reached, so the re-entry finds the approximation still $\bot$ and the term has no weak head normal form. Deciding it is where the ascent from $\bot$ of Section~\ref{sec:bridge} does its work. Appendix~\ref{app:omega-trace} is the run, generated by the interpreter.
  \fi

  \begin{figure}
    \centering
\begin{tikzpicture}[
  font=\footnotesize,
  state/.style={draw=red!65, fill=red!10, rounded corners, minimum width=12mm, minimum height=8mm, align=center, inner sep=2pt},
  back/.style={-Stealth, red!65, shorten >=1pt, shorten <=1pt},
]
  \node[state] (omega) at (0,0) {$\Omega$};
  \draw[back] (omega) to[out=40,in=-40,looseness=8] node[right]{\itshape $\beta$} (omega);
  \node[anchor=north, align=center] at (0,-1.3) {\footnotesize re-enters on the stack, no weak head normal form $\;\Longrightarrow\;$ $\bot$};
\end{tikzpicture}
    \caption{\bilingual{The graph of $\Omega = (\lambda x.\,x\,x)\,(\lambda x.\,x\,x)$, generated by running
      the interpreter: one state, whose head reduction contracts to itself.}{$\Omega = (\lambda x.\,x\,x)\,(\lambda x.\,x\,x)$ 的图,由运行解释器生成:一个状态,其头部归约收缩回它自身。}}
    \label{fig:omega}
  \end{figure}

  \section{\bilingual{Discussion}{讨论}}\label{sec:discussion}

  \subsection{\bilingual{The \texorpdfstring{$\lambda$}{lambda}-Calculus as a DSL for Graphs}{把 \texorpdfstring{$\lambda$}{lambda}-演算用作图的 DSL}}\label{sec:graphs}
  \ifPaperChinese
  解释器把一个项的展开折叠到它重复的状态上,于是一棵无限树成为一张有限图,而项的结构\emph{就是}那张图。因此,程序员如何编码一个算法,就决定了哪些子问题会重合,并随之决定图的大小以及它所耗的时间与空间;什么会折叠,取决于编码,而非指称。\footnote{零的环状流写成 $r = Y\,(\mathtt{cons}\,0)$ 时会折叠,其状态反复出现;但写成 $Y\,F\,\underline{0}$(其中 $F = \lambda s.\lambda n.\,\mathtt{cons}\,0\,(s\,(\mathtt{succ}\,n))$)时却不会:它的状态 $Y\,F\,\underline{0}, Y\,F\,\underline{1}, \dots$ 在结构上两两相异,尽管它们指称同一条流。}在整个有理片段上,即有限状态、正则、以及有限格上不动点的那些计算中,每个领域平时为管理一张图而手写的簿记,正是在状态同一性上做 tabling 所免费给出的同一样东西:动态规划的记忆表、博弈搜索的置换表、Datalog 可达性与指向分析的导出事实表,以及合一(unification)与递归类型的环检测相等,都是同一个机制,即对重复状态的识别;附录~\ref{app:dsl} 把每一个都做成解释器所运行的一个纯项。项的结构是这一成本上的一个杠杆,而它据以制表的相等性是另一个。于是,编译就是一种图变换,把源图映射到一张编译后的图,同一个解释器把它求解成一个到 Python 的自举编译器,它在解释器下用 61 秒、编译后用 12.5 秒编译完自身的源代码,峰值内存分别为 1.1 与 0.35 GB,在 CPython 3.11 上测得(附录~\ref{app:compiler});它据以制表的相等性,以及由此打开的相等性与归约的设计空间,我们在第~\ref{sec:design-space} 节展开。
  \else
  The interpreter folds a term's unfolding onto its repeated states, so an infinite tree becomes a finite graph and the term's structure \emph{is} that graph. How a programmer encodes an algorithm therefore fixes which subproblems coincide, and with them the size of the graph and the time and space it costs; what folds depends on the encoding, not the denotation.\footnote{The cyclic stream of zeros folds when written as $r = Y\,(\mathtt{cons}\,0)$, whose states recur, but not as $Y\,F\,\underline{0}$ with $F = \lambda s.\lambda n.\,\mathtt{cons}\,0\,(s\,(\mathtt{succ}\,n))$: its states $Y\,F\,\underline{0}, Y\,F\,\underline{1}, \dots$ are all structurally distinct, though they denote the same stream.} Across the rational fragment, the finite-state, regular, and finite-lattice-fixpoint computations, the bookkeeping each domain ordinarily writes by hand to manage a graph is exactly what tabling on state identity supplies for free: the memo table of dynamic programming, the transposition table of game search, the derived-fact table of Datalog reachability and points-to analysis, and the cycle-detecting equality of unification and recursive types are one mechanism, the recognition of a repeated state, and Appendix~\ref{app:dsl} works each as a pure term the interpreter runs. The structure of the term is one lever on this cost, and the equality it tables by is another. Compilation is then a graph transformation, mapping the source graph to a compiled one that the same interpreter solves into a bootstrap compiler to Python, which compiles its own source in 61 seconds under the interpreter and in 12.5 seconds compiled, with peak memory of 1.1 and 0.35 gigabytes, on CPython 3.11 (Appendix~\ref{app:compiler}); the equality it tables by, and the design space of equalities and reductions it opens, we take up in Section~\ref{sec:design-space}.
  \fi

  \subsection{\bilingual{Alternative Operational Semantics}{替代的操作语义}}\label{sec:design-space}
  \ifPaperChinese
  \textsf{Tabled} 折叠两个它能判定为相同的状态,而它据以折叠的可判定同一性,始终是状态的某个一阶编码的结构同一性;所选择的就是这个编码,即结构同一性\emph{是关于什么的}。解释器取 de Bruijn 项本身,即第~\ref{sec:bridge} 节的内化,也是这类同一性中最细的一种,由一次指针测试判定。编译器取标准的去函数化(defunctionalized)形式,以「函数体的散列值,加上按位置而非按绑定深度编号的自由变量」来命名每个闭包,于是仅在捕获绑定深度上不同的抽象会共享同一编码而被折叠到一起:还是结构同一性,只是如今针对一个更粗的表示,仍是常数时间的测试,并假定散列足够长、碰撞概率可以忽略(附录~\ref{app:compiler})。在一切编码之外是行为相等(behavioral equality),它折叠行为相同的状态、能达到整个有理片段,但在 $\lambda$-项上不可判定,因为判定一个项的树是否等于 $\Omega$ 的树,就会判定该项是否有弱头范式~\cite{barendregt1984-lambda-calculus};de Bruijn 同一性够不到整个有理片段,正如第~\ref{sec:graphs} 节的流 $Y\,F\,\underline{0}$ 所示。编译器的编码会折叠解释器所分开的状态:附录~\ref{app:compiler} 直接测量了这更高的命中率,在自编译自举上内化对象约少五倍,于是自举在 CPython 3.11 上随之快约五倍、内存只用其中一小部分。

  归约函数是第二个选择。\textsf{Tabled} 从任何呈现这种形式的方程的归约中读出一套树语义;弱头归约呈现的是 Lévy--Longo 树,而另一种归约则固定另一套惰性树语义,全都由同一个 \textsf{Tabled} 计算。于是纯 $\lambda$-演算不是一门语言,而是一族,以它所制表的相等性和它所运行的归约为索引,建立在单一的 tabled 求值之上。
  \else
  \textsf{Tabled} folds two states it can decide are the same, and the decidable identity it folds by is always the structural identity of some first-order encoding of the state; the choice is the encoding, structural identity \emph{of what}. The interpreter takes the de Bruijn term itself, the interning of Section~\ref{sec:bridge} and the finest such identity, decided by a pointer test. The compiler takes the standard defunctionalized form, naming each closure by a hash of its body with its free variables numbered by position rather than by binding depth, so abstractions differing only in the depth at which they bind a capture share an encoding and fold together: the same structural identity, now of a coarser representation, still a constant-time test, the hash assumed long enough for collisions to be negligible (Appendix~\ref{app:compiler}). Beyond every encoding lies behavioral equality, which folds states with equal behavior and attains the whole rational fragment but is undecidable on $\lambda$-terms, since deciding whether a term's tree equals that of $\Omega$ would decide whether the term has a weak head normal form~\cite{barendregt1984-lambda-calculus}; the de Bruijn identity falls short of the whole rational fragment, as the stream $Y\,F\,\underline{0}$ of Section~\ref{sec:graphs} shows. The compiler's encoding folds states the interpreter keeps apart: Appendix~\ref{app:compiler} measures the higher hit rate directly, about five times fewer interned objects on the self-compilation bootstrap, and on CPython 3.11 the bootstrap then runs about five times faster and in a fraction of the memory.

  The reduction function is the second choice. \textsf{Tabled} reads a tree semantics off whatever reduction presents an equation of this form; weak-head reduction presents the L\'evy--Longo tree, and a different reduction fixes a different lazy tree semantics, all computed by the same \textsf{Tabled}. The pure $\lambda$-calculus is thus not one language but a family, indexed by the equality it tables and the reduction it runs, over a single tabled evaluation.
  \fi

  \section{\bilingual{Related Work}{相关工作}}\label{sec:related}

  \paragraph{\bilingual{Domain-specific cycle detection.}{特定领域的环检测。}}
  \ifPaperChinese
  通过检测重入状态来让环状结构上的计算终止,这在若干领域早已各自独立地确立。逻辑编程中的 tabling~\cite{tamaki1986-tabled-resolution, chen1996-tabled-evaluation-delaying} 在逻辑程序结论算子的最小不动点~\cite{vanemden1976-predicate-logic-semantics}中补全递归调用的环,而非陷入循环。有理树与环状合一~\cite{colmerauer1984-trees, courcelle1983-infinite-trees} 处理带回边的一阶项。等递归(equirecursive)类型论通过跟踪已访问的类型对,使子类型与相等检查终止~\cite{amadio1993-subtyping-recursive-types, brandt1998-coinductive-type-equality}。二元决策图~\cite{bryant1986-bdd} 通过内化共享子图来把布尔函数规范化。这些方法各自把某一个领域的签名写死:Herbrand 项、一阶树、类型,或布尔函数。它们对数据并不参数化;每一种都是一门独立的方法,其环检测专属于它所处理的数据。
  \else
  Terminating computation over cyclic structures by detecting re-entrant states is long established in several domains independently. Tabling in logic programming~\cite{tamaki1986-tabled-resolution, chen1996-tabled-evaluation-delaying} completes cycles of recursive calls in the least fixpoint of a logic program's consequence operator~\cite{vanemden1976-predicate-logic-semantics} instead of looping. Rational trees and cyclic unification~\cite{colmerauer1984-trees, courcelle1983-infinite-trees} handle first-order terms with back-edges. Equirecursive type theory terminates subtyping and equality checks by tracking visited type pairs~\cite{amadio1993-subtyping-recursive-types, brandt1998-coinductive-type-equality}. Binary decision diagrams~\cite{bryant1986-bdd} canonicalize Boolean functions by interning shared sub-graphs. Each of these hard-wires one domain's signature: Herbrand terms, first-order trees, types, or Boolean functions. They are not parametric in the data; each is a separate discipline whose cycle detection is specific to the data it handles.
  \fi

  \paragraph{\bilingual{Sharing as an added construct.}{把共享作为附加构造。}}
  \ifPaperChinese
  图归约~\cite{wadsworth1971-semantics-lambda-calculus} 在可能含环的图上对项求值,但环只系结在递归是原语之处,即一个 $\mathtt{letrec}$ 或一个内置不动点;纯不动点组合子在每个 $\beta$-步都把其函数体复制一份,展开成一条无界的脊。带显式递归的环状 $\lambda$-项~\cite{ariola1997-explicit-recursion} 与 call-by-need~\cite{ariola1997-call-by-need} 把共享作为一个附加构造来提供。Morris 利用 $\mathtt{letrec}$ 来创建不可变的循环列表,以编程实现无限项合一~\cite{morris1978-list-structures-unification},展示了语言现有机制的一种应用。Hughes 的惰性记忆函数(lazy memo-functions)~\cite{hughes1985-lazy-memo-functions} 让环状结构保持成环:一个标记为 \texttt{memo} 的函数,当它再次被应用于存放在同一地址的实参时,复用其结果,因此在一个由 $\mathtt{letrec}$ 声明的环状列表上做记忆化的 $\mathtt{map}$ 会再次遇到同一个列表单元,并返回一个环状列表,该环经由它的记忆表闭合。该测试是指针同一性,因此结果是否成环取决于两个相等的实参是否是同一个对象,而代入一个在别处构建的相等实参会改变结果的结构。若改为测试实参是否相等,则需要每个实参的一个 hash-consing 副本,这使记忆函数完全严格,因而无法用于无限或环状的实参。一个在其自身结果仍在计算之中时再次出现的调用,则完全没有被处理。纯函数式语言若没有诸如可观察共享(observable sharing)~\cite{gill2009-observable-sharing} 之类的扩展,或一种把惰性系结的环穿过数据结构的环形程序写法~\cite{bird1984-circular-programs},甚至无法观察到共享。在所有这些之中,一个在环状列表上的普通递归函数,比如一个未做记忆化的 $\mathtt{map}$,都归约成一条无界的脊,而非一个环。没有一个能让演算保持不扩展:每一个都添加了某种演算原生没有的机制,即指针同一性、$\mathtt{letrec}$、一个 \texttt{memo} 标注,或一个可观察共享组合子。
  \else
  Graph reduction~\cite{wadsworth1971-semantics-lambda-calculus} evaluates terms over possibly cyclic graphs, but a cycle is tied only where recursion is primitive, a $\mathtt{letrec}$ or a built-in fixpoint; a pure fixpoint combinator has its body copied at each $\beta$-step and unfolds into an unbounded spine. Cyclic $\lambda$-terms with explicit recursion~\cite{ariola1997-explicit-recursion} and call-by-need~\cite{ariola1997-call-by-need} provide sharing as an added construct. Morris used $\mathtt{letrec}$ to create immutable cyclic lists to program unification over infinite terms~\cite{morris1978-list-structures-unification}, demonstrating one application of an existing language mechanism. Hughes's lazy memo-functions~\cite{hughes1985-lazy-memo-functions} keep cyclic structures cyclic: a function marked \texttt{memo} reuses its result when it is applied again to an argument stored at the same address, so a memoized $\mathtt{map}$ over a cyclic list declared by $\mathtt{letrec}$ meets the same list cell again and returns a cyclic list, the cycle closing through its memo table. The test is pointer identity, so whether the result is cyclic depends on whether two equal arguments are one object, and substituting an equal argument built elsewhere changes the structure of the result. Testing arguments for equality instead requires a hash-consed copy of each argument, which makes the memo-function completely strict and so unusable on infinite or cyclic arguments. A call that recurs while its own result is still being computed is not treated at all. A pure functional language cannot even observe sharing without an extension such as observable sharing~\cite{gill2009-observable-sharing} or a circular-program idiom that threads a lazily tied knot through a data structure~\cite{bird1984-circular-programs}. Across all of these, an ordinary recursive function over a cyclic list, such as an unmemoized $\mathtt{map}$, reduces to an unbounded spine rather than to a cycle. None leaves the calculus unextended: each adds a mechanism, pointer identity, $\mathtt{letrec}$, a \texttt{memo} annotation, or an observable-sharing combinator, that the calculus does not natively have.
  \fi

  \paragraph{\bilingual{Compilers as $\lambda$-terms}{作为 $\lambda$-项的编译器}}
  \ifPaperChinese
  Mogensen 的自解释器~\cite{mogensen1992-self-interpretation} 以高阶抽象语法(higher-order abstract syntax)表示一个项,因此被解释项中的代换就是解释器自身的 $\beta$-归约。Lynn 的自举编译器~\cite{lynn-compilers-for-contrarians} 把其源代码编译为一台组合子图归约机的汇编,这台机器的原语包括机器字算术、比较以及输入输出。lambda-8cc~\cite{ikuta2022-lambda-8cc} 是一个把 C 编译为 x86 的闭合无类型 $\lambda$-项,它的得来方式是:经由 ELVM 把 C 编译器 8cc 编译为 LambdaVM 的指令列表,LambdaVM 是一台以 $\lambda$-项写成的虚拟 CPU。因此它的 $\beta$-步运行的是那台 CPU,而不是直接归约编译器自身的定义:内存是一棵按地址的 24 个位索引的二叉树,每个字是一个 24 位的比特列表,于是每次读或写都要沿这棵树走至多 24 层,每次加法或比较都要逐位走过 24 个位。编译其文档中的示例程序需要几分钟到近一小时,以及几十到几百 GB 的内存,其作者推测在 $\lambda$-演算解释器中加入垃圾回收可以减少这些内存。
  \else
  Mogensen's self-interpreter~\cite{mogensen1992-self-interpretation} represents a term in higher-order abstract syntax, so substitution in the interpreted term is the interpreter's own $\beta$-reduction. Lynn's self-hosting compiler~\cite{lynn-compilers-for-contrarians} compiles its source to the assembly of a combinator graph-reduction machine whose primitives include machine-word arithmetic, comparison, and input and output. lambda-8cc~\cite{ikuta2022-lambda-8cc} is a closed untyped $\lambda$-term that compiles C to x86, obtained by compiling the C compiler 8cc through ELVM to the instruction list of LambdaVM, a virtual CPU written as a $\lambda$-term. Its $\beta$-steps therefore run that CPU rather than reduce the compiler's own definitions: memory is a binary tree indexed by the 24 bits of an address and every word is a list of 24 bits, so each load or store walks up to 24 levels of the tree and each addition or comparison steps through 24 bits. Compiling the example programs of its documentation takes minutes to an hour, and tens to hundreds of gigabytes of memory, which its author suspects a garbage collector in the $\lambda$-calculus interpreter would reduce.
  \fi


  \section{\bilingual{Conclusion}{结论}}\label{sec:conclusion}
  \ifPaperChinese
  一个 $\lambda$-项的每个子项都是一个状态,而对其弱头归约做 tabling,便把该项的含义算成一张有限的、可能含环的图;而标准做法则要伸到演算之外,取一个递归构造来造环、一处非纯来遍历它。这正是从两侧读到的纯性的代价:tabling 凭一种项自身未必能做出的识别来折叠环,因此这识别必须可判定、只能停在结构同一性,达不到完整的有理片段(第~\ref{sec:design-space} 节);而项这一侧被扣住不给,因此一个项无法把它自己的共享读回成一个 $\mathtt{letrec}$。

  解释一个项的那同一个 tabling 也编译它,而编译器本身就是它所求解的一个项(附录~\ref{app:compiler}):同一个对再次出现状态的可判定识别,把环折叠成一张有限图,并把没有弱头范式、且其求取访问有限多个项的项判定为 $\bot$。这些保证对解释器成立(第~\ref{sec:bridge} 节);编译后的形式以实验受检于同样的保证:在所测试的每个输入上,它都产生与解释执行相同的有限图。

  若作操作式的读法,一个 $\lambda$-项\emph{祈使}一场计算,一次归约接着一次;而 tabling 让它转而去\emph{声明}一场计算,只声明要构建什么样的图、如何转换它,其中值相等者可以自由替换,而对一个已经见过的状态的识别则落在求解器身上。这正是我们对 $\lambda$-演算所怀的理想:一种用来\emph{声明}计算、而非用来\emph{祈使}计算的语言。
  \else
  Each sub-term of a $\lambda$-term is a state, and tabling its weak-head reduction computes the term's meaning as a finite, possibly cyclic graph, where the standard account would reach outside the calculus for a recursion construct to build the cycle and an impurity to traverse it. This is the price of purity, read from two sides: tabling folds a cycle by a recognition the term may not perform, so the recognition must be decidable and stays structural identity, short of the full rational fragment (Section~\ref{sec:design-space}), while the term, denied it, cannot read its own sharing back out as a $\mathtt{letrec}$.

  The same tabling that interprets a term also compiles it, the compiler being itself a term it solves (Appendix~\ref{app:compiler}): one decidable recognition of a recurring state folds a cycle into a finite graph and decides a term with no weak head normal form whose demand visits finitely many terms as $\bot$. Those guarantees hold for the interpreter (Section~\ref{sec:bridge}); the compiled form is held to them experimentally, producing the same finite graph as the interpreted run on every input tested.

  Read operationally, a $\lambda$-term \emph{imperates} a computation, one reduction after another; tabling lets it \emph{declare} one instead, only what graph is built and how it is transformed, with equal values freely substitutable, while the recognition of a state already seen falls to the solver. This is the ideal we hold for the $\lambda$-calculus, a language in which a computation is \emph{declared}, not \emph{imperated}.
  \fi

  \fi 

  \ifTablambdaBody
  \ifTablambdaLNCS\bibliographystyle{splncs04}\else\bibliographystyle{ACM-Reference-Format}\fi
  \bibliography{references}
  \fi

  \ifTablambdaAppendix
  \appendix

  \section{\bilingual{The Bootstrap Compiler}{自举编译器}}\label{app:compiler}
  \ifPaperChinese
  我们的自举编译器本身就是一个纯 $\lambda$-项,解释器用解释项时所用的那同一套 tabled 归约在被引用(quoted)的源上求解它,于是程序分析被表达为求值。它把被引用的源程序(即表示为 Scott 编码数据的源程序)的每个子项,送到一棵 Python 抽象语法树中覆于该子项自身子项之上的一个节点,而它的解就是编译后的程序。编译策略是去函数化(defunctionalization)~\cite{reynolds1972-definitional-interpreters-defunctionalization},即把每个高阶值替换为一个具名的一阶数据构造子,并一致地应用于每个源子项:一个应用编译为一个 \texttt{Thunk},即一个携带其被调者与参数的挂起 redex;一个抽象编译为一个内化的 dataclass,其字段是它所捕获的自由变量,其 call 方法在一个参数上运行时,通过对编译后的函数体求值来执行 $\beta$-归约;一个变量编译为相应的参数或捕获字段。
  \else
  Our bootstrap compiler is itself a pure $\lambda$-term, which the interpreter solves on the quoted source by the same tabled reduction it runs to interpret a term, so program analysis is expressed as evaluation. It sends each sub-term of the quoted source program, represented as a Scott-encoded datum, to a node of a Python abstract syntax tree over that sub-term's own sub-terms, and its solution is the compiled program. The compilation strategy is defunctionalization~\cite{reynolds1972-definitional-interpreters-defunctionalization}, the replacement of each higher-order value by a named first-order data constructor, applied uniformly to every source sub-term: an application compiles to a \texttt{Thunk}, a suspended redex carrying its callee and argument; an abstraction compiles to an interned dataclass whose fields are the free variables it captures and whose call method, run on an argument, performs the $\beta$-reduction by evaluating the compiled body; a variable compiles to the corresponding argument or capture field.
  \fi

  \paragraph{\bilingual{The runtime and content-addressed closures}{运行时与内容寻址的闭包}}
  \ifPaperChinese
  运行时是极小的:一个 \texttt{Closure} 抽象基类、一个 \texttt{Thunk} 载体,以及一个内化算子;编译器发出的每个闭包类都是 \texttt{Closure} 的子类。一个 \texttt{Thunk} 的弱头范式由解释器用于尾调用链的那同一个记忆化最小不动点机制计算,因此 \texttt{Thunk} 内化镜像了项内化:一个再次出现的 redex 是同一个内化 \texttt{Thunk},只归约一次。环图折叠与把没有弱头范式、且其求取访问有限多个项的项判定为 $\bot$,对解释器成立(第~\ref{sec:bridge} 节);对编译后的程序,它们由测试检验:下文报告的断言核对编译执行在所测试的每个输入上产生与解释执行相同的有限图。闭包类是内容可寻址的:每个类以其编译后函数体的散列值命名,其捕获字段按位置而非按绑定深度编号,因此两个函数体形状相同、但在不同深度捕获变量的源抽象,编译为同一个类。类名携带其规范化后的类体的 SHA-256 摘要的 64 位;必须不碰撞的是互不相同的编译后类体——在自举最大的工件里约五千个——而不是表~\ref{tab:defun} 里的内化对象(编译后的自举运行中约一百万个),后者是按类与捕获为键的实例;即使类体多至其十倍,64 位名字发生生日碰撞的概率仍低于 $10^{-10}$,我们据此以散列相等判定编译后函数体相等。这就是第~\ref{sec:design-space} 节那个更粗、却仍可判定的同一性:它折叠解释器逐节点内化所分开的状态。
  \else
  The runtime is minimal: a \texttt{Closure} abstract base class, a \texttt{Thunk} carrier, and an interning operator, every closure class the compiler emits being a subclass of \texttt{Closure}. The weak-head normal form of a \texttt{Thunk} is computed by the same memoized least-fixpoint mechanism the interpreter uses for the tail-call chain, so \texttt{Thunk} interning mirrors term interning: a recurring redex is one interned \texttt{Thunk}, reduced once. The cyclic-graph folding and the decision as $\bot$ of terms with no weak head normal form whose demands visit finitely many terms hold for the interpreter (Section~\ref{sec:bridge}); for the compiled program they are tested, by the assertion, reported below, that the compiled run produces the same finite graph as the interpreted run on every input tested. Closure classes are content-addressable: each class is named by a hash of its compiled body, with its capture fields numbered by position rather than by binding depth, so two source abstractions whose bodies have the same shape but capture variables at different depths compile to one class. The name carries 64 bits of a SHA-256 digest of the class's canonicalized body, and what must not collide are the distinct compiled bodies, about five thousand in the bootstrap's largest artifact, not the interned objects of Table~\ref{tab:defun}, a million of them in the compiled bootstrap run, which are instances keyed by class and captures; the birthday probability of a collision among even ten times that many 64-bit names stays below $10^{-10}$, on which basis equal hashes decide equal compiled bodies. This is the coarser, still decidable, identity of Section~\ref{sec:design-space}: it folds states the interpreter's per-node interning keeps apart.
  \fi

  \paragraph{\bilingual{The self-compilation bootstrap}{自编译自举}}
  \ifPaperChinese
  编译器的第一个基准是把这套构造应用于它自身。编译器是一个纯 $\lambda$-项;解释器用同一套 tabling 在它自己被引用的源上求解它;输出是一个 Python 程序。一个纯 $\lambda$-项把一张图(源项)变换成另一张图(编译后的抽象语法树),既不借助附加的递归构造、也不观察指针同一性,这正是图变换栖身于演算之内的证据。表~\ref{tab:defun} 连同若干更轻的项一并报告这次自举,每一项都在 CPython 3.11 上以解释执行对照编译执行。编译形式在轻量项上大约快两倍半到近八倍,在自举上约快五倍,且自举的峰值内存约低三倍;而在两个最小的编译任务上,编译形式的峰值内存反而略高(表中内存比 0.83 与 0.89),已加载的运行时与闭包类的固定开销在该规模下尚未被摊销。这些时间与内存数字是单个解释器的一次快照、会随解释器而变,而制表对象的计数是确定的:编译形式在自举上物化出同样少五倍的制表对象。
  \else
  The compiler's first benchmark is the construction applied to itself. The compiler is a pure $\lambda$-term; the interpreter solves it on its own quoted source by the same tabling; the output is a Python program. That a pure $\lambda$-term transforms one graph, the source term, into another, the compiled abstract syntax tree, without an added recursion construct and without observing pointer identity, is the evidence that graph transformation lives inside the calculus. Table~\ref{tab:defun} reports this bootstrap together with lighter terms, each run interpreted against compiled on CPython 3.11. The compiled form is roughly two and a half to nearly eight times faster on the light terms and about five times faster on the bootstrap, with lower peak memory on the bootstrap by a factor of about three; on the two smallest compilations its peak memory is instead slightly higher (memory ratios 0.83 and 0.89 in the table), the fixed cost of the loaded runtime and closure classes not yet amortized at that scale. These time and memory figures are a snapshot of a single interpreter and shift across interpreters, whereas the tabled-object counts are deterministic: the compiled form materializes five times fewer tabled objects on the bootstrap.
  \fi

  \begin{table*}
    \caption{\bilingual{Interpreted versus compiled execution. Each row runs one application interpreted (interned
        term nodes) against compiled (interned \texttt{Thunk}s and closures). \emph{Speedup}, \emph{memory
    ratio}, and \emph{tabled ratio} compare interpreted to compiled; a ratio above one means the compiled form wins. Time and peak memory are a single measured snapshot on CPython 3.11, and the compiled column runs the program already compiled, so compilation time is excluded. The last row is the compiler compiling its own source.}{解释执行对比编译执行。每一行把一次应用以解释执行(内化的项节点)对照编译执行(内化的 \texttt{Thunk} 与闭包)。\emph{加速比}、\emph{内存比}、\emph{制表比}比较解释执行与编译执行;比值大于一意味着编译形式更优。时间与峰值内存是在 CPython 3.11 上单次测量的快照,而编译那一列运行的是已经编译好的程序,故编译时间被排除在外。最后一行是编译器编译它自己的源代码。}}
    \label{tab:defun}
    \centering\footnotesize
    \setlength{\tabcolsep}{4pt}
    \resizebox{\textwidth}{!}{
\begin{tabular}{lrrrrrrrrr}
\hline
App / input & \shortstack[r]{Interp\\time} & \shortstack[r]{Defun\\time} & Speedup & \shortstack[r]{Interp\\mem} & \shortstack[r]{Defun\\mem} & \shortstack[r]{Mem\\ratio} & \shortstack[r]{Interp\\tabled} & \shortstack[r]{Defun\\tabled} & \shortstack[r]{Tabled\\ratio} \\
\hline
  edit-distance kitten/sitting & 0.158 & 0.064 & 2.47 & 25 & 25 & 1.00 & 10636 & 6443 & 1.65 \\
  edit-distance intention/execution & 0.241 & 0.090 & 2.67 & 26 & 26 & 1.01 & 14747 & 8986 & 1.64 \\
  DEFUN on identity & 0.147 & 0.019 & 7.68 & 25 & 30 & 0.83 & 11754 & 1394 & 8.43 \\
  DEFUN on S & 0.273 & 0.043 & 6.40 & 28 & 31 & 0.89 & 22922 & 3919 & 5.85 \\
  DEFUN on DEFUN (bootstrap) & 61.168 & 12.502 & 4.89 & 1093 & 347 & 3.15 & 5003844 & 981693 & 5.10 \\
\hline
\end{tabular}
}
  \end{table*}

  \ifPaperChinese
  制表对象那几列统计每种形式所物化(materialize)的内化对象数,即整个运行触及的全部节点,库项与位级算术状态都在内,因此它们落实的是第~\ref{sec:application-sharing} 节代价界中的常数与低阶项,而非其 $(m{+}1)(n{+}1)$ 个子问题状态;编译形式物化得更少,约少一倍半到八倍多,在自举上约少五倍,这是更粗同一性带来的更高命中率。在解释执行的程序中,每个子项都是一个不同的内化节点,因此两个函数体形状相同、却在不同 de Bruijn 深度绑定其捕获的抽象是不同的节点。在编译执行的程序中,每个抽象是一个类的实例,该类以其函数体在位置式捕获字段上的散列值命名,因此那两个抽象共享同一个类,而当它们还闭合于相同的值时,其实例也重合。编译表示把这些结构等价的抽象坍缩为一个对象,而解释器的逐节点内化做不到;\texttt{Thunk} 内化又加重了这一效果,因为一个 \texttt{Thunk} 以其被调者与参数的同一性为键,而这两者已经被粗化。在所测试的每个输入上,编译执行产生与解释执行相同的有限图,这由随附的测试断言。
  \else
  The tabled-object columns count the interned objects each form materializes, the whole run's nodes, library terms and bit-level arithmetic states included, so they realize the constants and lower-order terms of the cost bound of Section~\ref{sec:application-sharing} rather than its $(m{+}1)(n{+}1)$ subproblem states; the compiled form materializes fewer, between about one and a half and more than eight times fewer, about five times fewer on the bootstrap, the higher hit rate of the coarser identity. In the interpreted program every sub-term is a distinct interned node, so two abstractions whose bodies have the same shape but bind their captures at different de Bruijn depths are distinct nodes. In the compiled program each abstraction is an instance of a class named by the hash of its body over positional capture fields, so those two abstractions share one class, and their instances coincide whenever they also close over the same values. The compiled representation collapses these structurally equivalent abstractions to one object, which the interpreter's per-node interning does not, and \texttt{Thunk} interning compounds the effect, since a \texttt{Thunk} is keyed by the identity of its callee and argument and those have already coarsened. The compiled run produces the same finite graph as the interpreted run on every input tested, asserted by the accompanying test.
  \fi

  \paragraph{\bilingual{What the calculus cannot read back}{演算无法读回的东西}}
  \ifPaperChinese
  自举表明,产生并变换一张环图是演算内部之事。与之互补的极限是从一张环图中把显式共享读回出来。求解器通过在求解器一侧识别出某个状态曾被见过,把一个递推折叠为一条回边,但一个项做不出这种识别:检测两个子项是否同一个节点是一次指针同一性测试,正是演算所禁止的那种非纯。这一不可能性按构造成立,而非有待证明:文法没有共享构造子、也没有同一性原语,项就是一棵有限树,而本文的每个语义函数(大步弱头求值、弱头归约的展开、Lévy--Longo 树)都是项的函数,因此一个项算出的任何东西都不可能依赖元语言如何共享其表示;能依赖的只有成本。因此一个纯项可以持有并遍历一张有限的环图,却无法把它序列化成一个带显式共享的表示,即一个 \texttt{letrec} 或一张具名节点的列表,因为它分辨不出一个共享节点与两份相等副本。编译器之所以能避开这个极限,只是因为它消费的是被引用的源,即一棵它按结构读取的有限树,而绝非一张图的指针结构。
  \else
  The bootstrap shows that producing and transforming a cyclic graph is internal to the calculus. The complementary limit is reading explicit sharing back out of one. The solver folds a recurrence into a back edge by recognizing, on the solver's side, that a state has been seen before, but a term cannot make that recognition: detecting that two sub-terms are the same node is a pointer-identity test, the very impurity the calculus forbids. The impossibility holds by construction rather than by proof: the grammar has no sharing constructor and no identity primitive, a term is a finite tree, and every semantic function of this paper, big-step weak-head evaluation, the unfolding of weak-head reduction, the L\'evy--Longo tree, is a function of the term, so nothing a term computes can depend on how the metalanguage shares its representation; only cost can. So a pure term can hold and traverse a finite cyclic graph yet cannot serialize it into a representation with explicit sharing, a $\mathtt{letrec}$ or a list of named nodes, because it cannot tell a shared node from two equal copies. The compiler escapes this limit only because it consumes the quoted source, a finite tree it reads structurally, never the pointer structure of a graph.
  \fi

  \paragraph{\bilingual{Compiled output for canonical terms}{若干典型项的编译输出}}\label{app:compiler-examples}
  \ifPaperChinese
  下面的列表是编译器对若干典型项的逐字输出,由机器生成、未经编辑;与上面的表~\ref{tab:defun} 一样,它们可由\anon[补充材料]{随附的代码~\cite{yang2026-tablambda}}重现。一个应用编译为一个 \texttt{Thunk},一个抽象编译为一个内化的类(其字段是它的捕获、其 call 方法是编译后的函数体),一个变量编译为一个参数或一个捕获字段。每个类以其函数体在位置式捕获字段上的散列值命名,因此一个在不同绑定深度被复用的项产生同一个类;在「不同深度的捕获」那个列表里,两个内层抽象从相差一层的深度捕获同一个变量,编译成同一个类,两处调用点都指名它。
  \else
  The listings below are the compiler's verbatim output for a few canonical terms, machine-generated and unedited; like Table~\ref{tab:defun} above, they can be reproduced from\anon[ the supplementary material]{the accompanying code~\cite{yang2026-tablambda}}. An application compiles to a \texttt{Thunk}, an abstraction to an interned class whose fields are its captures and whose call method is the compiled body, and a variable to an argument or a capture field. Each class is named by a hash of its body over positional capture fields, so a term reused at different binding depths produces one class; in the depth-differing-captures listing the two inner abstractions, capturing the same variable from depths one apart, compile to the one class both call sites name.
  \fi


\medskip\noindent\textbf{Identity.}\quad $\lambda x.\, x$
\begin{lstlisting}
@interned
class vg_5790a50754fc667f(Closure):

    def __call__(self, a):
        return a
compiled = vg_5790a50754fc667f()
\end{lstlisting}

\medskip\noindent\textbf{The constant combinator $K$.}\quad $\lambda x.\, \lambda y.\, x$
\begin{lstlisting}
@interned
class vg_267c5a013e0be552(Closure):
    cap_0: Closure

    def __call__(self, a):
        return self.cap_0

@interned
class vg_2f7005dfc64f22d3(Closure):

    def __call__(self, a):
        return vg_267c5a013e0be552(a)
compiled = vg_2f7005dfc64f22d3()
\end{lstlisting}

\medskip\noindent\textbf{Church numeral $2$.}\quad $\lambda x.\, \lambda y.\, x\, (x\, y)$
\begin{lstlisting}
@interned
class vg_1393e1124049719f(Closure):
    cap_0: Closure

    def __call__(self, a):
        return Thunk(self.cap_0, Thunk(self.cap_0, a))

@interned
class vg_1ca9b7d68a0fb084(Closure):

    def __call__(self, a):
        return vg_1393e1124049719f(a)
compiled = vg_1ca9b7d68a0fb084()
\end{lstlisting}

\medskip\noindent\textbf{Successor.}\quad $\lambda x.\, \lambda y.\, \lambda z.\, y\, (x\, y\, z)$
\begin{lstlisting}
@interned
class vg_504ccc190e93e6eb(Closure):
    cap_0: Closure

    def __call__(self, a):
        return vg_f03ad1e4ac91d757(a, self.cap_0)

@interned
class vg_956ac5818c92855f(Closure):

    def __call__(self, a):
        return vg_504ccc190e93e6eb(a)

@interned
class vg_f03ad1e4ac91d757(Closure):
    cap_0: Closure
    cap_1: Closure

    def __call__(self, a):
        return Thunk(self.cap_0, Thunk(Thunk(self.cap_1, self.cap_0), a))
compiled = vg_956ac5818c92855f()
\end{lstlisting}

\medskip\noindent\textbf{Addition.}\quad $\lambda x.\, \lambda y.\, \lambda z.\, \lambda w.\, x\, z\, (y\, z\, w)$
\begin{lstlisting}
@interned
class vg_09245224e9a7f871(Closure):

    def __call__(self, a):
        return vg_96d17350f4830bfe(a)

@interned
class vg_7dcf3312a46a289a(Closure):
    cap_0: Closure
    cap_1: Closure

    def __call__(self, a):
        return vg_d7831256c99dd51a(self.cap_0, a, self.cap_1)

@interned
class vg_96d17350f4830bfe(Closure):
    cap_0: Closure

    def __call__(self, a):
        return vg_7dcf3312a46a289a(self.cap_0, a)

@interned
class vg_d7831256c99dd51a(Closure):
    cap_0: Closure
    cap_1: Closure
    cap_2: Closure

    def __call__(self, a):
        return Thunk(Thunk(self.cap_0, self.cap_1), Thunk(Thunk(self.cap_2, self.cap_1), a))
compiled = vg_09245224e9a7f871()
\end{lstlisting}

\medskip\noindent\textbf{Multiplication.}\quad $\lambda x.\, \lambda y.\, \lambda z.\, x\, (y\, z)$
\begin{lstlisting}
@interned
class vg_68c11b23b73fbba6(Closure):
    cap_0: Closure

    def __call__(self, a):
        return vg_770fab1dc9b2bce6(self.cap_0, a)

@interned
class vg_770fab1dc9b2bce6(Closure):
    cap_0: Closure
    cap_1: Closure

    def __call__(self, a):
        return Thunk(self.cap_0, Thunk(self.cap_1, a))

@interned
class vg_f255f0e2d2822ac2(Closure):

    def __call__(self, a):
        return vg_68c11b23b73fbba6(a)
compiled = vg_f255f0e2d2822ac2()
\end{lstlisting}

\medskip\noindent\textbf{Zero test.}\quad $\lambda x.\, x\, (\lambda y.\, \lambda z.\, \lambda w.\, w)\, (\lambda y.\, \lambda z.\, y)$
\begin{lstlisting}
@interned
class vg_267c5a013e0be552(Closure):
    cap_0: Closure

    def __call__(self, a):
        return self.cap_0

@interned
class vg_2f7005dfc64f22d3(Closure):

    def __call__(self, a):
        return vg_267c5a013e0be552(a)

@interned
class vg_5790a50754fc667f(Closure):

    def __call__(self, a):
        return a

@interned
class vg_6e6aafd047ea9acf(Closure):

    def __call__(self, a):
        return vg_9dd83a0b38353a6c()

@interned
class vg_9dd83a0b38353a6c(Closure):

    def __call__(self, a):
        return vg_5790a50754fc667f()

@interned
class vg_bb1e0c653b7009af(Closure):

    def __call__(self, a):
        return Thunk(Thunk(a, vg_6e6aafd047ea9acf()), vg_2f7005dfc64f22d3())
compiled = vg_bb1e0c653b7009af()
\end{lstlisting}

\medskip\noindent\textbf{Self-application.}\quad $\lambda x.\, x\, x$
\begin{lstlisting}
@interned
class vg_1b49bc69f60c1e7c(Closure):

    def __call__(self, a):
        return Thunk(a, a)
compiled = vg_1b49bc69f60c1e7c()
\end{lstlisting}

\medskip\noindent\textbf{Depth-differing captures.}\quad $\lambda x.\, (\lambda y.\, x)\, (\lambda y.\, \lambda z.\, x)$
\begin{lstlisting}
@interned
class vg_267c5a013e0be552(Closure):
    cap_0: Closure

    def __call__(self, a):
        return self.cap_0

@interned
class vg_b96f08cff84409ba(Closure):
    cap_0: Closure

    def __call__(self, a):
        return vg_267c5a013e0be552(self.cap_0)

@interned
class vg_e2357023df13e3d1(Closure):

    def __call__(self, a):
        return Thunk(vg_267c5a013e0be552(a), vg_b96f08cff84409ba(a))
compiled = vg_e2357023df13e3d1()
\end{lstlisting}

  \section{\bilingual{Graph Computations as Pure Terms}{作为纯项的图计算}}\label{app:dsl}
  \ifPaperChinese
  下面每个例子都是一个纯 $\lambda$-项;该项到达有限多个状态,每个状态都经由一个会终止的归约到达,因此解释器把它渲染为一张有限图;每个例子都作为一个测试再现\anon[~(补充材料)]{~\cite{yang2026-tablambda}}:测试用解释器求值该项,并断言其结果,即读回的值(从图上读出的布尔值或数码)或共享/成环节点的内化同一性,而没有弱头范式的项的结果就是求解器的 $\bot$ 自身。布尔值是 Church 编码的,$\mathtt{T} = \lambda a.\lambda b.\,a$ 与 $\mathtt{F} = \lambda a.\lambda b.\,b$,并有 $\mathtt{or} = \lambda p.\lambda q.\,p\,\mathtt{T}\,q$ 与 $\mathtt{and} = \lambda p.\lambda q.\,p\,q\,\mathtt{F}$,而 $\mathtt{cons}$ 与 $Y$ 同第~\ref{sec:bridge} 节。每个例子都局限于一个可处理的核心:正则词法分析而非一般的上下文无关解析,有理树合一而非一般的最一般合一子,有限状态模型检验而非无限状态。
  \else
  Each example below is a pure $\lambda$-term that reaches finitely many states, each by a terminating reduction, so the interpreter renders it as a finite graph; each is reproduced as a test\anon[~(supplementary material)]{~\cite{yang2026-tablambda}}: the test evaluates the term with the interpreter and asserts its outcome, the read-back value (a Boolean or numeral read off the graph) or the interned identity of a shared or cyclic node, the outcome for a term with no weak head normal form being the solver's $\bot$ itself. The Booleans are Church, $\mathtt{T} = \lambda a.\lambda b.\,a$ and $\mathtt{F} = \lambda a.\lambda b.\,b$, with $\mathtt{or} = \lambda p.\lambda q.\,p\,\mathtt{T}\,q$ and $\mathtt{and} = \lambda p.\lambda q.\,p\,q\,\mathtt{F}$, and $\mathtt{cons}$ and $Y$ are as in Section~\ref{sec:bridge}. Each example is confined to a tractable core: regular lexing but not general context-free parsing, rational-tree unification but not a general most-general unifier, finite-state model checking but not infinite-state.
  \fi

  \paragraph{\bilingual{Mapping a cyclic list}{映射一个环状列表}}
  \ifPaperChinese 普通的 $\mathtt{map}$,其中没有任何具备环感知的东西, \else The ordinary $\mathtt{map}$, with nothing cycle-aware in it, \fi
  \[
    \mathtt{map}\,f = Y\,\big(\lambda\mathit{self}.\,\lambda\mathit{lst}.\;
    \mathit{lst}\,(\lambda h.\lambda t.\,\mathtt{cons}\,(f\,h)\,(\mathit{self}\,t))\;\mathtt{nil}\big),
  \]
  \ifPaperChinese 作用于 $r = Y\,(\mathtt{cons}\,0)$ 时折叠成一个有限的环形列表,即一个持有 $f\,0$ 的单元;该单元的尾部指向它自身;在有限列表上它就是教科书里的 map。 \else applied to $r = Y\,(\mathtt{cons}\,0)$ folds to a finite circular list, one cell holding $f\,0$ whose tail points to itself; on a finite list it is the textbook map. \fi

  \paragraph{\bilingual{Dynamic programming}{动态规划}}
  \ifPaperChinese 一棵深度为 $n$、其节点共享两个子节点的完美二叉树,是一张有 $n+1$ 个状态的图;该图展开为 $2^{n}$ 个叶子: \else A perfect binary tree of depth $n$ whose nodes share both children is a graph of $n+1$ states that unfolds to $2^{n}$ leaves: \fi
  \[
    \mathtt{node} = \lambda l.\lambda r.\lambda m.\lambda f.\,m\,l\,r, \qquad
    \mathtt{leaf} = \lambda v.\lambda m.\lambda f.\,f\,v,
  \]
  \begin{gather*}
    \mathtt{any} = Y\,\big(\lambda\mathit{self}.\lambda t.\;
    t\,(\lambda l.\lambda r.\,\mathtt{or}\,(\mathit{self}\,l)\,(\mathit{self}\,r))\,(\lambda v.\,v)\big),\\
    t_0 = \mathtt{leaf}\;\mathtt{F},\quad t_{k+1} = \mathtt{node}\;t_k\;t_k,
  \end{gather*}
  \ifPaperChinese 因此 $\mathtt{any}\,t_n$ 读出为 $\mathtt{F}$,其 $n+1$ 个共享子树各求值一次,在 $n$ 上是线性的,而朴素递归会访问 $2^{n}$ 个叶子。 \else so $\mathtt{any}\,t_n$ reads to $\mathtt{F}$ with each of the $n+1$ shared subtrees evaluated once, linear in $n$ where a naive recursion would visit $2^{n}$ leaves. \fi

  \paragraph{\bilingual{Game search}{博弈搜索}}
  \ifPaperChinese
  极小化极大是一棵与或树:一个极大化局面是其各走法之上的析取,一个极小化局面则是合取,一个终局是一个布尔值;一个\emph{置换}(transposition),即由不同走法次序到达的同一个局面,是同一个内化状态,只计算一次。
  \else
  Minimax is an and-or tree: a maximizing position is the disjunction over its moves, a minimizing position the conjunction, a terminal a Boolean; a \emph{transposition}, the same position reached by a different move order, is the same interned state, computed once.
  \fi
  \begin{gather*}
    \mathtt{max} = \lambda l.\lambda r.\lambda x.\lambda n.\lambda f.\,x\,l\,r, \quad
    \mathtt{min} = \lambda l.\lambda r.\lambda x.\lambda n.\lambda f.\,n\,l\,r,\\
    \mathtt{leaf} = \lambda v.\lambda x.\lambda n.\lambda f.\,f\,v,
  \end{gather*}
  \begin{multline*}
    \mathtt{value} = Y\,\big(\lambda\mathit{self}.\lambda p.\;
      p\,(\lambda l.\lambda r.\,\mathtt{or}\,(\mathit{self}\,l)(\mathit{self}\,r))\\
    (\lambda l.\lambda r.\,\mathtt{and}\,(\mathit{self}\,l)(\mathit{self}\,r))\,(\lambda v.\,v)\big).
  \end{multline*}

  \paragraph{\bilingual{Graph reachability and program analysis}{图可达性与程序分析}}
  \ifPaperChinese 有向图上的可达性(含环)是直接结论算子的最小不动点~\cite{vanemden1976-predicate-logic-semantics}。Andersen 式的指向(别名)分析~\cite{andersen1994-program-analysis-c, bravenboer2009-doop},即由 $\mathtt{assign}(p,q)$ 与 $\mathtt{pointsTo}(q,o)$ 推出 $\mathtt{pointsTo}(p,o)$,是单调 Datalog。一个 Herbrand 基有 $k$ 个原子的 ground 程序就是一个纯项:模型是 $k$ 个布尔值的 Church 元组,直接结论算子 $T_P$ 是一个元组到元组的项(每个原子一个析取,遍历推出它的各子句体的合取),而最小 Herbrand 模型是把 $T_P$ 从全假元组出发迭代 $k$ 次,即一次 Church 数码迭代,$k$ 界住格的高度;查询一个目标原子是一次投影,而推导环从不成为求解器的环:界住运行的是迭代次数,不是重入。认出一个重访的状态需要一次同一性测试,而演算没有同一性;单调算子在有限格上的稳定只需要一个计数(新事实还能出现多少次),而计数演算有,即一个数码。环状的输入由此成为无环的计算:闭合环的那条边,不过是一个原子下一轮的真值咨询另一个原子这一轮的真值。此处 tabling 的收益就是导出事实表本身:每一轮的模型是一个内化的项,其各胞元只被计算一次,无论有多少子句体咨询它们;而普通归约每咨询一次就重算一次。环状图 $a\to b\to c\to a$ 加上 $c\to d$ 与孤立的 $e$,以及针对 $a=\mathtt{new}\,o_1;\,b=a;\,c=b$、对象为 $o_1, o_2$ 的指向程序,是以下子句集 \else Reachability over a directed graph, cycles included, is the least fixpoint of the immediate-consequence operator~\cite{vanemden1976-predicate-logic-semantics}. Andersen-style points-to (alias) analysis \cite{andersen1994-program-analysis-c, bravenboer2009-doop}, $\mathtt{pointsTo}(p,o)$ from $\mathtt{assign}(p,q)$ and $\mathtt{pointsTo}(q,o)$, is monotone Datalog. A ground program over a Herbrand base of $k$ atoms is one pure term: a model is a Church tuple of $k$ Booleans, the consequence operator $T_P$ is a tuple-to-tuple term, one disjunction per atom over the conjunctions of its clauses' bodies, and the least Herbrand model is $T_P$ iterated $k$ times from the all-false tuple, a Church-numeral iteration, $k$ bounding the lattice's height; a goal atom is a projection, and a derivation cycle never becomes a solver cycle, the iteration count rather than a re-entry bounding the run. Recognizing a revisited state takes an identity test, which the calculus lacks; stabilization of a monotone operator on a finite lattice takes only a count of how many times a new fact can appear, and a count the calculus has, as a numeral. The cyclic input thereby becomes an acyclic computation: the edge closing the cycle is no more than one atom's next-round truth consulting another's current-round truth. The tabling win here is the derived-fact table itself: each round's model is one interned term whose cells are computed once, however many clause bodies consult them, where plain reduction recomputes a cell at every consultation. The cyclic graph $a\to b\to c\to a$ with $c\to d$ and an isolated $e$, and the points-to program for $a=\mathtt{new}\,o_1;\,b=a;\,c=b$ over the objects $o_1, o_2$, are the clause sets \fi
  \[
    \begin{aligned}
      &\mathtt{reach}(a),\quad \mathtt{reach}(b)\leftarrow\mathtt{reach}(a),\quad
      \mathtt{reach}(c)\leftarrow\mathtt{reach}(b),\\
      &\mathtt{reach}(a)\leftarrow\mathtt{reach}(c),\quad \mathtt{reach}(d)\leftarrow\mathtt{reach}(c);\\
      &\mathtt{pt}(a,o_1),\quad \mathtt{pt}(b,X)\leftarrow\mathtt{pt}(a,X),\quad
      \mathtt{pt}(c,X)\leftarrow\mathtt{pt}(b,X),
    \end{aligned}
  \]
  \ifPaperChinese
  其 Herbrand 基包含不可推导的原子 $\mathtt{reach}(e)$ 与 $\mathtt{pt}(c,o_2)$;读回断言 $\mathtt{reach}(d)$ 与 $\mathtt{pt}(c,o_1)$ 成立,而那两个不成立,读出的是 FALSE 而非 $\bot$。
  \else
  whose Herbrand bases include the underivable atoms $\mathtt{reach}(e)$ and $\mathtt{pt}(c,o_2)$; the readback asserts that $\mathtt{reach}(d)$ and $\mathtt{pt}(c,o_1)$ hold while those two do not, reading FALSE rather than $\bot$.
  \fi

  \paragraph{\bilingual{Lexical analysis}{词法分析}}
  \ifPaperChinese 一个确定性有限自动机是一张有限的环状转移图,而运行它是对输入的一次折叠;该折叠把状态穿引其中。把偶状态与符号 $a$ 编码为 $\mathtt{T}$,奇状态与 $b$ 编码为 $\mathtt{F}$,二者都是双向选择子, \else A deterministic finite automaton is a finite cyclic transition graph, and running it is a fold over the input that threads the state. Encoding the even state and the symbol $a$ as $\mathtt{T}$, the odd state and $b$ as $\mathtt{F}$, both two-way selectors, \fi
  \[
    \delta = \lambda s.\lambda x.\,s\,(x\,\mathtt{F}\,\mathtt{T})\,(x\,\mathtt{T}\,\mathtt{F}), \qquad
    \mathtt{accepts} = \lambda s.\,s\,\mathtt{T}\,\mathtt{F},
  \]
  \[
    \mathtt{run} = Y\,\big(\lambda\mathit{self}.\lambda s.\lambda w.\;
    w\,(\lambda h.\lambda t.\,\mathit{self}\,(\delta\,s\,h)\,t)\,s\big),
  \]
  \ifPaperChinese 而 $\mathtt{accepts}\,(\mathtt{run}\,\mathtt{T}\,w)$ 为 $\mathtt{T}$ 当且仅当 $w$ 中 $a$ 的个数为偶数。 \else and $\mathtt{accepts}\,(\mathtt{run}\,\mathtt{T}\,w)$ is $\mathtt{T}$ exactly when $w$ has an even number of $a$s. \fi

  \paragraph{\bilingual{Unification and recursive types}{合一与递归类型}}
  \ifPaperChinese
  两张解图所展开成的行为之间的相等,即有理树合一~\cite{colmerauer1984-trees} 与等递归($\mu$-)类型相等~\cite{amadio1993-subtyping-recursive-types, brandt1998-coinductive-type-equality} 的可判定核心,是余归纳的~\cite{courcelle1983-infinite-trees}:一次重访到自身的比较未曾发现分歧,应当回答相等,这是一个最大不动点;而求值器算的是最小不动点,于是裸布尔相等在相等的环上停在 $\bot$。给比较加上护卫便补上这个缺口:裁决流 $\mathtt{eqS}\,a\,b = \mathtt{cons}\,(\mathtt{eq}\,h_a\,h_b)\,(\mathtt{eqS}\,t_a\,t_b)$ 在递归之前先暴露一层裁决,于是求解器像折叠零的流(第~\ref{sec:application-cycles} 节)那样折叠它。在 $Y\,(\mathtt{cons}\;0)$ 对 $Y\,(\lambda s.\,\mathtt{cons}\;0\;s)$ 上,裁决流是单个 $\mathtt{T}$ 构成的有限环;对 $Y\,(\mathtt{cons}\;1)$,第一层裁决就是 $\mathtt{F}$。「没有 $\mathtt{F}$ 出现」从有限的裁决图上读出,这与画出任何一张图是同一种元语言观察。
  \else
  Equality of the behaviors two solution graphs unfold to, the decidable core of rational-tree unification \cite{colmerauer1984-trees} and equirecursive ($\mu$-) type equality \cite{amadio1993-subtyping-recursive-types, brandt1998-coinductive-type-equality}, is coinductive~\cite{courcelle1983-infinite-trees}: a comparison that revisits itself has found no disagreement and should answer equal, a greatest fixpoint, while the evaluator computes least ones, so a bare Boolean equality stalls at $\bot$ on an equal cycle. Guarding the comparison closes the gap: the verdict stream $\mathtt{eqS}\,a\,b = \mathtt{cons}\,(\mathtt{eq}\,h_a\,h_b)\,(\mathtt{eqS}\,t_a\,t_b)$ exposes one per-level verdict before recursing, so the solver folds it as it folds the stream of zeros (Section~\ref{sec:application-cycles}). On $Y\,(\mathtt{cons}\;0)$ against $Y\,(\lambda s.\,\mathtt{cons}\;0\;s)$ the verdict stream is the finite cycle of a single $\mathtt{T}$; against $Y\,(\mathtt{cons}\;1)$ the first verdict is $\mathtt{F}$. That no $\mathtt{F}$ occurs is read off the finite verdict graph, the same metalanguage observation that draws every graph.
  \fi

  \paragraph{\bilingual{The metalanguage needs tabling too}{元语言也需要 tabling}}
  \ifPaperChinese
  解释器自身的宿主语言项构造器反身地说明了这一点。通过复用一个构造器来写出一个共享子项,会得到一个项;该项的构造把共享指数式地展开,而一个环状的构造器则不会终止——直到该构造器按绑定子深度被记忆化,这是 tabling 的宿主语言对应物。直觉的代码恰恰在元语言不做 tabling 之处变慢或不终止。这里的备忘键是句法构造器,而非演算据以折叠的一阶行为同一性,但道理是一样的。
  \else
  The interpreter's own host-language term builder illustrates the point reflexively. Writing a shared subterm by reusing a builder gives a term whose construction unfolds the sharing exponentially, and a cyclic builder would not terminate, until the builder is memoized by binder depth, the host-language analogue of tabling. Intuitive code is slow or non-terminating exactly where the metalanguage does not table. The memo key here is the syntactic builder rather than the first-order behavioral identity the calculus folds on, but the moral is the same.
  \fi

  \paragraph{\bilingual{Further instances}{更多实例}}
  \ifPaperChinese
  我们预期同一个有理核心也支撑着此处未予构造的其他例子:e-图上的相等饱和(equality saturation)~\cite{willsey2021-egg, zhang2023-egglog} 与二元决策图~\cite{bryant1986-bdd},二者都是对一张共享图的内化;环状堆上的垃圾回收可达性;以及作为有限抽象域上不动点的抽象解释。每一个都是有限状态、正则、或有限格上不动点的计算,即该构造所读取的那种形状。
  \else
  We expect the same rational core to underlie others not constructed here: equality saturation over e-graphs
  \cite{willsey2021-egg, zhang2023-egglog} and binary decision diagrams \cite{bryant1986-bdd}, both of which are interning of a shared graph; garbage-collection reachability over a cyclic heap; and abstract interpretation as a fixpoint over a finite abstract domain. Each is a finite-state, regular, or finite-lattice-fixpoint computation, the shape the construction reads.
  \fi

  \section{\bilingual{Worked Examples}{实例}}\label{app:examples}
  \ifPaperChinese
  第~\ref{sec:application} 节三个子节背后的生成材料:每个项连同它由以构建的库,以及求解器的一段 trace。所有这些都由实现生成、而非手工誊写,trace 由运行解释器得到、列表由源项得到,且可由\anon[补充材料]{随附的代码~\cite{yang2026-tablambda}}重现。
  \else
  The generated material behind the three subsections of Section~\ref{sec:application}: each term with the library it is built from, and a trace of the solver. All of it is generated from the implementation, not transcribed by hand, the traces by running the interpreter and the listings from the source terms, and can be reproduced from\anon[ the supplementary material]{the accompanying code~\cite{yang2026-tablambda}}.
  \fi

  \subsection{\bilingual{Edit Distance: The Term and Its Library}{编辑距离:项及其库}}\label{app:editdistance-code}
  \ifPaperChinese
  编辑距离项 $\mathtt{ed}$ 及其由以构建的纯 lambda 库。每条定义只用到它上方的名字,并由源项渲染,因此该列表不会与实现脱节:绑定变量名就是源所携带的那些,而一个库常量按名字显示、而非展开。递归是 $Y$ 组合子,字符码是 BinNat,而 $\mathtt{nil}$ 也是 BinNat 零,因此整个列表是闭合的纯 $\lambda$-演算。
  \else
  The edit-distance term $\mathtt{ed}$ and the pure-lambda library it is built from. Each definition uses only the names above it, and is rendered from the source term, so the listing cannot drift from the implementation: the bound-variable names are the ones the source carries, and a library constant is shown by name rather than expanded. Recursion is the $Y$ combinator, the character codes are BinNats, and $\mathtt{nil}$ is also the BinNat zero, so the whole listing is closed pure $\lambda$-calculus.
  \fi

\smallskip\noindent\textit{fixed point}\par\nobreak
\begin{align*}
\begin{autobreak}
Y = \lambda f.\,
(\lambda x.\,
f\,
(x\,
x))\,
(\lambda x.\,
f\,
(x\,
x))
\end{autobreak}
\end{align*}
\smallskip\noindent\textit{Scott booleans}\par\nobreak
\begin{align*}
\begin{autobreak}
\mathtt{true} = \lambda a.\,
\lambda b.\,
a
\end{autobreak}
\end{align*}
\begin{align*}
\begin{autobreak}
\mathtt{false} = \lambda a.\,
\lambda b.\,
b
\end{autobreak}
\end{align*}
\begin{align*}
\begin{autobreak}
\mathtt{and} = \lambda p.\,
\lambda q.\,
p\,
q\,
\mathtt{false}
\end{autobreak}
\end{align*}
\begin{align*}
\begin{autobreak}
\mathtt{or} = \lambda p.\,
\lambda q.\,
p\,
\mathtt{true}\,
q
\end{autobreak}
\end{align*}
\smallskip\noindent\textit{Scott lists (nil is also the BinNat zero)}\par\nobreak
\begin{align*}
\begin{autobreak}
\mathtt{nil} = \lambda c.\,
\lambda n.\,
n
\end{autobreak}
\end{align*}
\begin{align*}
\begin{autobreak}
\mathtt{cons} = \lambda h.\,
\lambda t.\,
\lambda c.\,
\lambda n.\,
c\,
h\,
t
\end{autobreak}
\end{align*}
\smallskip\noindent\textit{comparison verdicts (a verdict picks one of less, equal, greater)}\par\nobreak
\begin{align*}
\begin{autobreak}
\mathtt{less} = \lambda \mathit{less}.\,
\lambda \mathit{equal}.\,
\lambda \mathit{greater}.\,
\mathit{less}
\end{autobreak}
\end{align*}
\begin{align*}
\begin{autobreak}
\mathtt{equal} = \lambda \mathit{less}.\,
\lambda \mathit{equal}.\,
\lambda \mathit{greater}.\,
\mathit{equal}
\end{autobreak}
\end{align*}
\begin{align*}
\begin{autobreak}
\mathtt{greater} = \lambda \mathit{less}.\,
\lambda \mathit{equal}.\,
\lambda \mathit{greater}.\,
\mathit{greater}
\end{autobreak}
\end{align*}
\smallskip\noindent\textit{BinNat bit operations}\par\nobreak
\begin{align*}
\begin{autobreak}
\mathtt{one} = \mathtt{cons}\,
\mathtt{true}\,
\mathtt{nil}
\end{autobreak}
\end{align*}
\begin{align*}
\begin{autobreak}
\mathtt{not} = \lambda \mathit{bit}.\,
\mathit{bit}\,
\mathtt{false}\,
\mathtt{true}
\end{autobreak}
\end{align*}
\begin{align*}
\begin{autobreak}
\mathtt{xor} = \lambda \mathit{left}.\,
\lambda \mathit{right}.\,
\mathit{left}\,
(\mathtt{not}\,
\mathit{right})\,
\mathit{right}
\end{autobreak}
\end{align*}
\begin{align*}
\begin{autobreak}
\mathtt{maj} = \lambda \mathit{first}.\,
\lambda \mathit{second}.\,
\lambda \mathit{third}.\,
\mathtt{or}\,
(\mathtt{and}\,
\mathit{first}\,
\mathit{second})\,
(\mathtt{or}\,
(\mathtt{and}\,
\mathit{first}\,
\mathit{third})\,
(\mathtt{and}\,
\mathit{second}\,
\mathit{third}))
\end{autobreak}
\end{align*}
\begin{align*}
\begin{autobreak}
\mathtt{biteq} = \lambda \mathit{left}.\,
\lambda \mathit{right}.\,
\mathit{left}\,
\mathit{right}\,
(\mathtt{not}\,
\mathit{right})
\end{autobreak}
\end{align*}
\begin{align*}
\begin{autobreak}
\mathtt{bitcmp} = \lambda x.\,
\lambda y.\,
\mathtt{biteq}\,
x\,
y\,
\mathtt{equal}\,
(x\,
\mathtt{greater}\,
\mathtt{less})
\end{autobreak}
\end{align*}
\smallskip\noindent\textit{BinNat arithmetic and comparison}\par\nobreak
\begin{align*}
\begin{autobreak}
\mathtt{iszero} = Y\,
(\lambda \mathit{iszero}.\,
\lambda n.\,
n\,
(\lambda \mathit{bit}.\,
\lambda \mathit{rest}.\,
\mathit{bit}\,
\mathtt{false}\,
(\mathit{iszero}\,
\mathit{rest}))\,
\mathtt{true})
\end{autobreak}
\end{align*}
\begin{align*}
\begin{autobreak}
\mathtt{addc} = Y\,
(\lambda \mathit{addc}.\,
\lambda \mathit{carry}.\,
\lambda a.\,
\lambda b.\,
a\,
(\lambda x.\,
\lambda \mathit{xs}.\,
b\,
(\lambda y.\,
\lambda \mathit{ys}.\,
\mathtt{cons}\,
(\mathtt{xor}\,
(\mathtt{xor}\,
x\,
y)\,
\mathit{carry})\,
(\mathit{addc}\,
(\mathtt{maj}\,
x\,
y\,
\mathit{carry})\,
\mathit{xs}\,
\mathit{ys}))\,
(\mathtt{cons}\,
(\mathtt{xor}\,
x\,
\mathit{carry})\,
(\mathit{addc}\,
(\mathtt{and}\,
x\,
\mathit{carry})\,
\mathit{xs}\,
\mathtt{nil})))\,
(b\,
(\lambda y.\,
\lambda \mathit{ys}.\,
\mathtt{cons}\,
(\mathtt{xor}\,
y\,
\mathit{carry})\,
(\mathit{addc}\,
(\mathtt{and}\,
y\,
\mathit{carry})\,
\mathtt{nil}\,
\mathit{ys}))\,
(\mathit{carry}\,
\mathtt{one}\,
\mathtt{nil})))
\end{autobreak}
\end{align*}
\begin{align*}
\begin{autobreak}
\mathtt{add} = \lambda a.\,
\lambda b.\,
\mathtt{addc}\,
\mathtt{false}\,
a\,
b
\end{autobreak}
\end{align*}
\begin{align*}
\begin{autobreak}
\mathtt{succ} = \lambda n.\,
\mathtt{add}\,
n\,
\mathtt{one}
\end{autobreak}
\end{align*}
\begin{align*}
\begin{autobreak}
\mathtt{cmp} = Y\,
(\lambda \mathit{cmp}.\,
\lambda a.\,
\lambda b.\,
a\,
(\lambda x.\,
\lambda \mathit{xs}.\,
b\,
(\lambda y.\,
\lambda \mathit{ys}.\,
\mathit{cmp}\,
\mathit{xs}\,
\mathit{ys}\,
\mathtt{less}\,
(\mathtt{bitcmp}\,
x\,
y)\,
\mathtt{greater})\,
(\mathtt{iszero}\,
(\mathtt{cons}\,
x\,
\mathit{xs})\,
\mathtt{equal}\,
\mathtt{greater}))\,
(b\,
(\lambda y.\,
\lambda \mathit{ys}.\,
\mathtt{iszero}\,
(\mathtt{cons}\,
y\,
\mathit{ys})\,
\mathtt{equal}\,
\mathtt{less})\,
\mathtt{equal}))
\end{autobreak}
\end{align*}
\begin{align*}
\begin{autobreak}
\mathtt{eq} = \lambda a.\,
\lambda b.\,
\mathtt{cmp}\,
a\,
b\,
\mathtt{false}\,
\mathtt{true}\,
\mathtt{false}
\end{autobreak}
\end{align*}
\begin{align*}
\begin{autobreak}
\mathtt{min} = \lambda a.\,
\lambda b.\,
\mathtt{cmp}\,
a\,
b\,
a\,
a\,
b
\end{autobreak}
\end{align*}
\smallskip\noindent\textit{edit distance}\par\nobreak
\begin{align*}
\begin{autobreak}
\mathtt{len} = Y\,
(\lambda \mathit{len}.\,
\lambda \mathit{items}.\,
\mathit{items}\,
(\lambda \mathit{head}.\,
\lambda \mathit{tail}.\,
\mathtt{succ}\,
(\mathit{len}\,
\mathit{tail}))\,
\mathtt{nil})
\end{autobreak}
\end{align*}
\begin{align*}
\begin{autobreak}
\mathtt{ed} = Y\,
(\lambda \mathit{ed}.\,
\lambda a.\,
\lambda b.\,
a\,
(\lambda \mathit{head}_{a}.\,
\lambda \mathit{tail}_{a}.\,
b\,
(\lambda \mathit{head}_{b}.\,
\lambda \mathit{tail}_{b}.\,
\mathtt{eq}\,
\mathit{head}_{a}\,
\mathit{head}_{b}\,
(\mathit{ed}\,
\mathit{tail}_{a}\,
\mathit{tail}_{b})\,
(\mathtt{succ}\,
(\mathtt{min}\,
(\mathit{ed}\,
\mathit{tail}_{a}\,
\mathit{tail}_{b})\,
(\mathtt{min}\,
(\mathit{ed}\,
\mathit{tail}_{a}\,
b)\,
(\mathit{ed}\,
a\,
\mathit{tail}_{b})))))\,
(\mathtt{len}\,
a))\,
(\mathtt{len}\,
b))
\end{autobreak}
\end{align*}

  \subsection{\bilingual{Edit Distance: The Trace}{编辑距离:trace}}\label{app:editdistance-trace}
  \ifPaperChinese
  算法~\ref{alg:solvewhnf} 的求解器在 $\mathtt{ed}\,(\mathtt{cons}\,a\,(\mathtt{cons}\,b\,\mathtt{nil}))\allowbreak\,(\mathtt{cons}\,c\,(\mathtt{cons}\,d\,\mathtt{nil}))$ 上的完整运行,由对解释器插桩记录而来。每一行是求解器到达的一个子问题:第一次是 \texttt{compute},此时它求解该项并把其弱头范式写入表;之后每次出现是 \texttt{hit},此时该项是同一个内化节点,其制表的弱头范式被直接返回而无需再次求解。缩进是调用深度,而 $a,b,c,d$ 是 BinNat 字符码。九个子问题被计算,即动态规划表的 $(m{+}1)(n{+}1)$ 个状态,而其余每一次出现都是一次命中。一行 \texttt{hit} 记录的是对一个制表条目的一次读取,而非图~\ref{fig:editdistance} 的一条边。图中的十二条边正是 trace 中互不相同的调用对,即哪个调用者请求哪个子问题:其中八条作为根以下的 \texttt{compute} 行出现,即每个子问题的首次到达;四条作为某个调用对的首次 \texttt{hit} 出现,即进入高亮节点的额外入边。其余八行 \texttt{hit} 重复一个已经出现过的调用对:把一次调用的三个结果组合起来的后继与三路最小运算,会再次读取同样的制表条目,而每一次读取都被记录。
  \else
  The complete run of the solver of Algorithm~\ref{alg:solvewhnf} on $\mathtt{ed}\,(\mathtt{cons}\,a\,(\mathtt{cons}\,b\,\mathtt{nil}))\allowbreak\,(\mathtt{cons}\,c\,(\mathtt{cons}\,d\,\mathtt{nil}))$, recorded by instrumenting the interpreter. Each line is a subproblem the solver reaches: \texttt{compute} the first time, when it solves the term and writes its weak head normal form to the table, and \texttt{hit} on every later occurrence, when the term is the same interned node and its tabled weak head normal form is returned without solving it again. Indentation is the call depth, and $a,b,c,d$ are the BinNat character codes. Nine subproblems are computed, the $(m{+}1)(n{+}1)$ states of the dynamic-programming table, and every other occurrence is a hit. A \texttt{hit} line records a read of a tabled entry, not an edge of Figure~\ref{fig:editdistance}. The figure's twelve edges are the distinct calling pairs of the trace, which caller requests which subproblem: eight appear as the \texttt{compute} lines below the root, the first reach of each subproblem, and four as the first \texttt{hit} of a pair, the further edges into the highlighted nodes. The remaining eight \texttt{hit} lines repeat a pair that has already appeared: the successor and three-way minimum that combine a call's three results read the same tabled entries again, and every read is logged.
  \fi

\begin{lstlisting}[language=]
compute ed (cons a (cons b nil)) (cons c (cons d nil)) = 2
  compute ed (cons b nil) (cons d nil) = 1
    compute ed nil nil = 0
    compute ed nil (cons d nil) = 1
    compute ed (cons b nil) nil = 1
    hit ed nil (cons d nil) = 1
    hit ed nil nil = 0
  compute ed (cons b nil) (cons c (cons d nil)) = 2
    hit ed nil (cons d nil) = 1
    compute ed nil (cons c (cons d nil)) = 2
    hit ed (cons b nil) (cons d nil) = 1
    hit ed (cons b nil) (cons d nil) = 1
    hit ed nil (cons d nil) = 1
  compute ed (cons a (cons b nil)) (cons d nil) = 2
    hit ed (cons b nil) nil = 1
    hit ed (cons b nil) (cons d nil) = 1
    compute ed (cons a (cons b nil)) nil = 2
    hit ed (cons b nil) (cons d nil) = 1
    hit ed (cons b nil) nil = 1
  hit ed (cons b nil) (cons c (cons d nil)) = 2
  hit ed (cons b nil) (cons d nil) = 1
\end{lstlisting}

  \subsection{\bilingual{The Stream of Zeros: The Term and Its Library}{零的流:项及其库}}\label{app:cyclic-zeros-code}
  \ifPaperChinese
  环状流 $r = Y\,(\mathtt{cons}\,0)$ 及其由以构建的小型库,由源项渲染。元素 $0$ 是 Church 零,$\mathtt{cons}$ 是 Scott 构造子,而递归是 $Y$ 组合子,因此整个列表是闭合的纯 $\lambda$-演算。
  \else
  The cyclic stream $r = Y\,(\mathtt{cons}\,0)$ and the small library it is built from, rendered from the source terms. The element $0$ is the Church zero, $\mathtt{cons}$ is the Scott constructor, and recursion is the $Y$ combinator, so the whole listing is closed pure $\lambda$-calculus.
  \fi

\smallskip\noindent\textit{fixed point}\par\nobreak
\begin{align*}
\begin{autobreak}
Y = \lambda f.\,
(\lambda x.\,
f\,
(x\,
x))\,
(\lambda x.\,
f\,
(x\,
x))
\end{autobreak}
\end{align*}
\smallskip\noindent\textit{Scott list constructor}\par\nobreak
\begin{align*}
\begin{autobreak}
\mathtt{cons} = \lambda h.\,
\lambda t.\,
\lambda c.\,
\lambda n.\,
c\,
h\,
t
\end{autobreak}
\end{align*}
\smallskip\noindent\textit{the element}\par\nobreak
\begin{align*}
\begin{autobreak}
0 = \lambda s.\,
\lambda z.\,
z
\end{autobreak}
\end{align*}
\smallskip\noindent\textit{the cyclic stream}\par\nobreak
\begin{align*}
\begin{autobreak}
r = Y\,
(\mathtt{cons}\,
0)
\end{autobreak}
\end{align*}

  \subsection{\bilingual{The Stream of Zeros: The Trace}{零的流:trace}}\label{app:cyclic-zeros-trace}
  \ifPaperChinese
  对 $r = Y\,(\mathtt{cons}\,0)$ 的图的枚举,由解释器记录;枚举本身是元语言观察(第~\ref{sec:bridge} 节)。一个状态上的每次 \texttt{TabledWHNF} 向解释器求取该状态的弱头范式,即运行弱头方程 \texttt{WHNF}、收缩头部 redex 以暴露 cons 单元,然后跟随尾部;记 $W = \lambda x.\,\mathtt{cons}\,0\,(x\,x)$,展开 $Y$ 把 $r$ 改写成自应用 $W\,W$;该 $W\,W$ 暴露出 $\mathtt{cons}\,0\,(W\,W)$。尾部 $W\,W$ 在该节点仍处于枚举路径上时被到达,于是回边在那里闭合。每一步都从真实的内化项读出,所得的弱头范式与解释器的结果核对一致。
  \else
  The enumeration of the graph of $r = Y\,(\mathtt{cons}\,0)$, recorded from the interpreter; the enumeration itself is metalanguage observation (Section~\ref{sec:bridge}). Each \texttt{TabledWHNF} on a state demands its weak head normal form from the interpreter, running the weak-head equation \texttt{WHNF} that contracts head redexes to expose the cons cell, then follows the tail; writing $W = \lambda x.\,\mathtt{cons}\,0\,(x\,x)$, unfolding $Y$ rewrites $r$ to the self-application $W\,W$, which exposes $\mathtt{cons}\,0\,(W\,W)$. The tail $W\,W$ is reached while that node is still on the enumeration's path, so the back edge closes there. Every step is read off the real interned terms, the weak head normal forms checked against the interpreter's.
  \fi

\begin{lstlisting}[language=]
TabledWHNF r:
  WHNF:  Y (cons 0)  ->  W W  ->  cons 0 (W W)
  compute r => cons 0 (W W)
  TabledWHNF W W:
    WHNF:  W W  ->  cons 0 (W W)
    compute W W => cons 0 (W W)
    tail W W: on the stack -> back edge (cycle closes)
\end{lstlisting}

  \subsection{\bilingual{$\Omega$: The Term and Its Library}{$\Omega$:项及其库}}\label{app:omega-code}
  \ifPaperChinese
  发散项 $\Omega = \omega\,\omega$ 及其由以构建的自应用 $\omega = \lambda x.\,x\,x$,由源项渲染。
  \else
  The divergent term $\Omega = \omega\,\omega$ and the self-application $\omega = \lambda x.\,x\,x$ it is built from, rendered from the source terms.
  \fi

\smallskip\noindent\textit{self-application}\par\nobreak
\begin{align*}
\begin{autobreak}
\omega = \lambda x.\,
x\,
x
\end{autobreak}
\end{align*}
\smallskip\noindent\textit{the divergent term}\par\nobreak
\begin{align*}
\begin{autobreak}
\Omega = \omega\,
\omega
\end{autobreak}
\end{align*}

  \subsection{\bilingual{$\Omega$: The Trace}{$\Omega$:trace}}\label{app:omega-trace}
  \ifPaperChinese
  求解器在 $\Omega$ 上的运行,由解释器记录,与零的流那次相平行。\texttt{TabledWHNF} 运行弱头方程 \texttt{WHNF},它收缩头部 redex;收缩结果是 $\Omega$ 本身,即同一个内化项,因此求解器在它仍处于栈上、未到达任何弱头范式时重入它,并返回其运行中逼近 $\bot$;重新计算不改变它,因此迭代已收敛,求解器停机。该归约从真实的项读出,收缩结果核对为 $\Omega$,结果核对为 $\bot$。
  \else
  The solver's run on $\Omega$, recorded from the interpreter and parallel to the stream's. \texttt{TabledWHNF} runs the weak-head equation \texttt{WHNF}, which contracts the head redex; the contractum is $\Omega$ itself, the same interned term, so the solver re-enters it while it is still on the stack with no weak head normal form reached and returns its running approximation, $\bot$; recomputation leaves it unchanged, so the iteration has converged and the solver halts. The reduction is read off the real terms, the contractum checked to be $\Omega$ and the result $\bot$.
  \fi

\begin{lstlisting}[language=]
TabledWHNF Omega:
  WHNF:  (\x. x x) (\x. x x)  ->  (\x. x x) (\x. x x)   (head redex contracts to Omega itself, the same interned term)
  re-enter Omega: on the stack, no weak head normal form  ->  bottom
Omega = bottom
\end{lstlisting}

  \fi 

  \ifTablambdaBody\else
  \ifTablambdaLNCS\bibliographystyle{splncs04}\else\bibliographystyle{ACM-Reference-Format}\fi
  \bibliography{references}
  \fi

  \ifPaperChinese\clearpage
\end{CJK*}\fi
\end{document}